\documentclass[11pt,letterpaper]{article}

\usepackage{typearea}
\paperwidth 8.5in \paperheight 11in \typearea{16}

\usepackage{theorem,latexsym,graphicx}
\usepackage{amsmath,amssymb,enumerate}
\usepackage{xspace}
\usepackage{bm}
\usepackage{ifpdf}
\usepackage{color}
\usepackage[compact]{titlesec}
\usepackage{algorithm, algorithmic}
\usepackage{paralist}

\newcommand{\ignore}[1]{}

\definecolor{Darkblue}{rgb}{0,0,0.4}
\definecolor{Brown}{cmyk}{0,0.81,1.,0.60}
\definecolor{Purple}{cmyk}{0.45,0.86,0,0}

\newcommand{\mydriver}{hypertex} \ifpdf
 \renewcommand{\mydriver}{pdftex}
\fi
\usepackage[breaklinks,\mydriver]{hyperref}
\hypersetup{colorlinks=true,%pdfborder={1 1 1 [3]},%
            citebordercolor={.6 .6 .6},linkbordercolor={.6 .6 .6},%
citecolor=Darkblue,urlcolor=black,linkcolor=Darkblue,pagecolor=black}

\newcommand{\lref}[2][]{\hyperref[#2]{#1~\ref*{#2}}}


\newtheorem{theorem}{Theorem}[section]
\newtheorem{definition}[theorem]{Definition}

\newtheorem{lemma}[theorem]{Lemma}
\newtheorem{claim}[theorem]{Claim}

\allowdisplaybreaks

\newenvironment{Myquote}{\par\begingroup
\addtolength{\leftskip}{1em} \rightskip\leftskip }{\par
\endgroup
}

\newenvironment{proof}{

\noindent{\bf Proof:}} {\hfill$\blacksquare$

}

\newenvironment{proofof}[1]{

\noindent{\bf Proof of {#1}:}} {\hfill$\blacksquare$

}

\newcommand{\junk}[1]{}

\newcommand{\R}[0]{{\ensuremath{\mathbb{R}}}}

\def\floor#1{\lfloor #1 \rfloor}
\def\ceil#1{\lceil #1 \rceil}

\def\flow {\ensuremath{f}\xspace}

\def\tf{\mathcal{T}}
\def\a{\ensuremath{\mathcal{A}}\xspace}

\def\cg{{\sf Cong}}
\def\g{\mathcal{G}}
\def\n{\mathcal{N}}

\def\f{\ensuremath {\mathcal{F}}\xspace}

\def\opt{{\sf Opt}}

\def\cov{\ensuremath{\Pi}\xspace}
\def\rcov{{\sf Robust(\cov)}\xspace}
\def\robkcov{\ensuremath{{\sf Robust}}_k(\cov)\xspace}
\def\mm{{\sf MaxMin}\xspace}
\def\auga{{\sf Augment}\xspace}
\def\fst{\ensuremath {\Phi_T}\xspace}
\def\snd{\ensuremath {{\sf Augment}_T}\xspace}

\newcommand{\sse}{\subseteq}

\newcommand{\ts}{\textstyle}

\newcounter{note}[section]

\newcommand{\initOneLiners}{%
    \setlength{\itemsep}{0pt}
    \setlength{\parsep }{0pt}
    \setlength{\topsep }{0pt}
}
\newenvironment{OneLiners}[1][\ensuremath{\bullet}]
    {\begin{list}
        {#1}
        {\initOneLiners}}
    {\end{list}}

\newcommand{\Tstar}{\ensuremath{T^*}\xspace}
\newcommand{\Phistar}{\ensuremath{\Phi^*}\xspace}
\def\calT{\mathcal{T}}

\usepackage{times}

\begin{document}

\title{Thresholded Covering Algorithms for \\ Robust and Max-Min
  Optimization\thanks{An extended abstract containing the results of
    this paper and of~\cite{GNR-rob-gen} appeared jointly in
    \emph{Proceedings of the 37th International Colloquium on Automata,
      Languages and Programming (ICALP), 2010}.}}

\author{
Anupam Gupta\thanks{Computer Science Department, Carnegie Mellon
    University, Pittsburgh, PA 15213, USA. Supported in part by
    NSF awards CCF-0448095 and CCF-0729022, and an Alfred P.~Sloan
    Fellowship. Email: anupamg@cs.cmu.edu}
\and Viswanath Nagarajan\thanks{IBM T.J. Watson Research Center, Yorktown Heights, NY 10598, USA. Email:
viswanath@us.ibm.com} \and R. Ravi\thanks{Tepper School of Business, Carnegie Mellon University,
  Pittsburgh, PA 15213, USA. Supported in part by NSF grant
  CCF-0728841. Email: ravi@cmu.edu}
}
\date{}

%\begin{titlepage}
%\def\thepage{}
%\thispagestyle{empty}

\maketitle
\begin{abstract}
  The general problem of robust optimization is this: one of several
  possible scenarios will appear tomorrow and require to be covered, but things are more expensive
  tomorrow than they are today. What should you anticipatorily buy
  today, so that the worst-case cost (summed over both days) is
  minimized?
We consider the \emph{$k$-robust model}
%introduced by Feige et al.~\cite{FJMM07},
where the possible outcomes tomorrow
are given by all demand-subsets of size $k$.
%For example, in a set cover instance, if any one of the   $\binom{n}{k}$
%subsets of the universe that have size $k$ may appear  tomorrow, what is a good course of action?

%Feige et   al.~\cite{FJMM07}, and later, Khandekar et al.~\cite{KKMS08},  considered this \emph{$k$-robust model} where the possible outcomes
%  tomorrow are given by all demand-subsets of size $k$, and gave  algorithms for the set cover problem, and the Steiner tree and
%  facility location problems in this model, respectively.

  \medskip In this paper, we give the following simple and intuitive
  template for $k$-robust problems: \emph{having built some anticipatory
    solution, if there exists a single demand whose augmentation cost is
    larger than some threshold (which is $\approx \opt/k$), augment the
    anticipatory solution to cover this demand as well, and repeat.}  We show that this template gives improved approximation
  algorithms for $k$-robust Steiner tree and set cover,
%  (improving on the performance ratios of the previously known  algorithms, by a logarithmic factor in the case of set cover),
and present the first approximation algorithms for $k$-robust Steiner
  forest, minimum-cut and multicut.  Our main technical contribution lies in proving certain net-type properties
  for these covering problems, which are based on dual-rounding and primal-dual ideas; these properties might be of
  some independent interest. All the approximation ratios
  (except for multicut) are nearly optimal.

%In addition to its   apparent simplicity and efficacy, a salient feature of our framework  is the interesting set of technical ideas---in particular,
%  dual-rounding ideas---needed to show good performance of this heuristic.

  \medskip As a by-product of our techniques, we get algorithms for
  max-min problems of the form: ``\emph{given a covering problem instance,
    which $k$ of the elements are costliest to cover?}'' If the covering
  problem does not naturally define a submodular function, very little
  is known about these problems. For the problems mentioned above,
%  (set cover, multicut, Steiner forest),
we show that their $k$-max-min
  versions have performance guarantees similar to those for the
  $k$-robust problems.
%, which are optimal in many cases. E.g., for set   cover, we give an $O(\log m + \log n)$ approximation, which improves
%  the previous result by a logarithmic factor, and nearly matches the  $\Omega(\frac{\log m}{\log \log m} + \log n)$ hardness.

\end{abstract}

%\end{titlepage}

\section{Introduction}

Consider the following \emph{$k$-robust} set cover problem: we are given a set system $(U, \mathcal{F} \sse 2^U)$.
Tomorrow some set of $k$ elements $S \sse U$ will want to be covered; however, today we don't know what this set will
be. One strategy is to wait until tomorrow and buy an $O(\log n)$-approximate set cover for this set. However, sets are
cheaper today: they will cost $\lambda$ times as much tomorrow as they cost today. Hence, it may make sense to buy some
anticipatory partial solution today (i.e. in the first-stage), and then complete it tomorrow (i.e. second-stage) once
we know the actual members of the set $S$. Since we do not know anything about the set $S$ (or maybe we are
risk-averse), we want to plan for the worst-case, and minimize:
\[ (\text{cost of anticipatory solution}) + \lambda \cdot \max_{S: |S|
  \leq k} ( \text{additional cost to cover } S).
\]
Early approximation results for robust problems~\cite{DGRS05,GGR06} had assumed that the collection of possible sets
$S$ was explicitly given (and the performance guarantee depended logarithmically on the size of this collection). Since
this seemed quite restrictive, Feige et al.~\cite{FJMM07} proposed the $k$-robust model where any of the $\binom{n}{k}$
subsets $S$ of size $k$ could arrive. Though this collection of possible sets was potentially exponential sized (for
large values of $k$), the hope was to get results that did not depend polynomially on $k$.

For the $k$-robust {\em set cover} problem, Feige et al.~\cite{FJMM07} gave an $O(\log m \log n)$-approximation
algorithm using the online algorithm for set cover within an LP-rounding-based algorithm (\`a la~\cite{SS04}). They
also showed $k$-robust set cover to be $\Omega(\smash{\frac{\log m}{\log\log m}} + \log n)$ hard---which left a
logarithmic gap between the upper and lower bounds.  However, an online algorithm based approach is unlikely to close
this gap, since the online algorithm for set cover is necessarily a log-factor worse that its offline
counterparts~\cite{AAABN03}.

Closely related to the $k$-robust model are {\em $k$-max-min} problems, where given a covering problem instance the
goal is to determine the $k$-set of demands that are costliest to cover. Indeed,~\cite{FJMM07} used $k$-max-min set
cover as a subroutine in their algorithm for $k$-robust set cover. If the covering problem defines a submodular
objective then the $k$-max-min problem can be solved via constrained submodular optimization. However, most natural
covering problems do not yield submodular functions, and so previous results cannot be applied directly. For
$k$-max-min set cover~\cite{FJMM07} used an online algorithm to obtain an $O(\log m \log n)$-approximation algorithm.
%As shown in the companion paper~\cite{GNR-rob-gen},

Apart from improving these   results in context of set cover, one may want to develop algorithms for other $k$-robust
and $k$-max-min problems.  E.g., for the \emph{$k$-robust min-cut} problem, some set $S$ of $k$ sources will want to be
separated from the sink vertex tomorrow, and we want to find the best way to cut edges to minimize the total cost
incurred (over the two days) for the worst-case $k$-set $S$. Similarly, in the \emph{$k$-max-min Steiner forest}, we
are given a metric space with a collection of source-sink pairs, and seek $k$ source-sink pairs that incur the maximum
(Steiner forest) connection cost. Although the online-based-framework~\cite{FJMM07} can be extended to give algorithms
for other $k$-max-min problems, it does not yield approximation guarantees better than the (deterministic) online
competitive ratios. Moreover, for $k$-robust problems other than set cover, the LP-rounding framework in~\cite{FJMM07}
does not extend directly; this obstacle was also observed by Khandekar et al.~\cite{KKMS08} who studied $k$-robust
Steiner tree and facility location.
%, who gave combinatorial constant-factor approximation algorithms for $k$-robust Steiner tree and facility location.

%In this paper

\subsection{Main Results}

In this paper, we present a general template to design algorithms for $k$-robust and $k$-max-min problems.  We go
beyond the online-based approach and obtain tighter approximation ratios that nearly match the offline guarantees; see
the table below. We improve on previous results, by obtaining an $O(\log m + \log n)$ factor for $k$-robust set cover,
and improving the constant in the approximation factor for Steiner tree. We also give the first algorithms for some
other standard covering problems, getting constant-factor approximations for both $k$-robust Steiner forest---which was
left open by Khandekar et al.---and for $k$-robust min-cut, and an $O(\smash{\frac{\log^2 n}{\log\log n}})$
approximation for $k$-robust multicut. Our algorithms do not use a max-min subroutine directly: however, our approach
ends up giving us approximation algorithms for $k$-max-min versions of set cover, Steiner forest, min-cut and multicut;
all but the one for multicut are best possible under standard assumptions.
%; the ones for Steiner forest and multicut are the first known algorithms for their $k$-max-min versions.

%\subsection{Cover the Needy Framework and Techniques}
%\subsection{Our Techniques}

An important contribution of our work
%(even more than the new/improved approximation guarantees)
is the simplicity of the algorithms, and the
ideas in their analysis.  The following is our actual algorithm for
$k$-robust set cover.
\begin{quote}
  Suppose we ``guess'' that the maximum second-stage cost in the optimal
  solution is $T$.  Let $A \sse U$ be all elements for which the
  cheapest set covering them costs more than $\beta \cdot T/k$, where
  $\beta = O(\log m + \log n)$.  We build a set cover on $A$ as our
  first stage. (Say this cover costs~$C_T$.)
\end{quote}
To remove the guessing, try all values of $T$ and choose the solution that incurs the least total cost $C_T + \lambda
\beta T$. Clearly, by design, no matter which $k$ elements arrive tomorrow, it will not cost us more than $\lambda
\cdot k \cdot \beta T/k = \lambda \beta T$ to cover them, which is within $\beta$ of what the optimal solution pays.
This guess-and-verify framework is formalized in Sections~\ref{sec:framework} and~\ref{sec:mm-framework}.

% \emph{Why is the cost $C_T$ close to the optimum?}
The key step of our analysis is to argue why $C_T$ is close to optimum. We briefly describe the intuition; details
appear in \lref[Section]{sec:set-cover}.  Suppose $C_T \gg \beta \opt$: then the fractional solution to the LP for set
cover for $A$ would cost~$\gg \frac{\beta}{\ln n} \opt \geq \opt$, and so would its dual. Our key technical
contribution is to show how to {\em ``round'' this dual LP} to find a ``witness'' $A' \sse A$ with only $k$ elements,
and also a corresponding feasible dual of value $\gg \opt$---i.e., the dual value does not decrease much in the
rounding. This step uses the fact that each element in $A$ is expensive to cover individually. Using duality again,
this proves that the optimal LP value, and hence the optimal set cover for these $k$ elements $A'$, would cost much
more than $\opt$---a contradiction!

In fact, our algorithms for the other $k$-robust problems are almost identical to this one; indeed, the only slightly
involved algorithm is that for $k$-robust Steiner forest. Of course, the proofs to bound the cost $C_T$ need different
ideas in each case. These involve establishing certain {\em net-type properties} for the respective covering problems
(which imply the existence of such a witness $A' \sse A$ of size $k$), and represent our main technical contribution.
The proofs for set cover, min-cut and multicut are based on {\em dual-rounding}. In the case of Steiner forest,
directly rounding the dual is difficult, and we give a {\em primal-dual} argument.
% The min-cut problem exposes an
% issue that we ignored in the above intuition for Set Cover: $\opt$ is
% the sum of anticipatory costs plus $\lambda \times$ the worst-case
% recourse cost---to do the charging, we have to ensure that the
% contradiction set $A'$

For the cut-problems, one has to deal with additional issues because $\opt$ consists of two stages that have to be
charged to separately, and this requires a careful Gomory-Hu-tree-based charging. Even after this, we have to show the
following net-type property: if the cut for a set of sources $A$ is large (costs $\gg \opt$) and each source in $A$ has
a high individual cut ($\gg \opt/k$) then there is a witness $A' \sse A$ of at most $k$ sources for which the cut is
also large ($\gg \opt$).
%This requires aggregating the flows (i.e. dual-rounding for cut problems).
To this end, we prove new flow-aggregation lemmas for single-sink flows using Steiner-tree-packing results, and for
multiflows using oblivious routing~\cite{R08}; both proofs are possibly
of independent interest.%  (e.g., the latter gives an algorithm for the
% all-or-nothing multicommodity flows when the min-cuts are large)

%\agnote{This last bit is not very punchy, nor is it very clear---I've
%  just put in the soundbites. Can we make this clearer.}

% As mentioned above, our algorithms for $k$-robust problems extend to
% $k$-max-min problems as well. One can think of the $k$-max-min set cover
% problem as asking: \emph{which $k$-set would be costliest to cover?} In
% fact, if one thinks of $\lambda = 1$, the problems are the same

To get a quick overview of our basic approach, see the analysis for Steiner tree in \lref[Appendix]{sec:steiner-tree}.
While the result is simple and does not require rounding the dual, it is a nice example of our framework in action. In
\lref[Section]{sec:notation} we present the formal framework for $k$-robust and $k$-max-min problems, and abstract out
the properties that we'd like from our algorithms. Then \lref[Section]{sec:set-cover} contains such an algorithm for
$k$-robust set cover---Min-cut, Steiner forest and multicut appear in \lref[Sections]{sec:minimum-cut},
\ref{sec:steiner-forest} and~\ref{sec:robust-multicut}. The table below summarizes the best-known approximation ratios
for various covering problems in the offline, $k$-robust and online models (results denoted $*$ are in the present
paper).

\begin{center}
\begin{small}
\begin{tabular}{|c|c|c|c|} \hline
{\bf Problem} & {\bf Offline }& {\bf $k$-robust } & {\bf Deterministic Online }\\
\hline\hline
Set Cover & $\ln n$ & $O(\log m + \log n)$ \quad ($*$) & $O(\log m\cdot \log n)$~\cite{AAABN03} \\
&$(1-o(1))\ln n$~\cite{F98} & $\Omega\left(\log n + \frac{\log m}{\log \log m}\right)$~\cite{FJMM07} & $\Omega\left(\frac{\log m\cdot \log n}{\log \log m +\log \log n}\right)$~\cite{AAABN03} \\
\hline Steiner Tree& 1.39~\cite{BGRS10} &  4.5 \quad ($*$) & $\Theta(\log n)$~\cite{IM91} \\
%&  &  & \\
\hline Steiner Forest& 2~\cite{AKR95,GW95} &  10 \quad ($*$) & $\Theta(\log n)$~\cite{BC97} \\
\hline Minimum Cut & 1 & 17 \quad ($*$)& $O(\log^3 n\cdot \log\log n)$~\cite{AAABN04,HHR03} \\
\hline Multicut & $O(\log n)$~\cite{GVY96} & $O\left(\frac{\log^2n}{\log\log n}\right)$ \quad ($*$) & $O(\log^3 n\cdot  \log\log n)$~\cite{AAABN04,HHR03}\\
\hline
%\caption{.}
\end{tabular}
\end{small}
\end{center}

%The algorithm for $k$-robust multicut is in \lref[Appendix]{sec:robust-multicut}.
%We present a general reduction from robust problems to the corresponding max-min
%problems in~\lref[Appendix]{sec:gen-sets}. In that section, we also describe extensions of our work to more general
%uncertainty sets, e.g. incorporating matroid and knapsack type constraints.

\subsection{Related Work}
\label{sec:related-work}

Approximation algorithms for robust optimization was initiated by Dhamdhere et al.~\cite{DGRS05}: they study the case
when the scenarios were explicitly listed, and gave constant-factor approximations for Steiner tree and facility
location, and logarithmic approximations to mincut/multicut problems.  Golovin et al.~\cite{GGR06} improved the mincut
result to a constant factor approximation, and also gave an $O(1)$-approximation for robust shortest-paths. The
algorithms in~\cite{GGR06} were also ``thresholded algorithms'' and the algorithms in this paper can be seen as natural
extensions of that idea to more complex uncertainty sets and larger class of problems (the uncertainty set
in~\cite{GGR06} only contained singleton demands).
% The~\cite{GGR06} method
% for minimum cut uses a simple guess-and-prune strategy that guesses the
% second stage cut cost and cuts off in the first stage all terminals
% whose stand-alone min-cut costs substantially more than the guessed
% value. Our methods in this paper are a natural extension of this idea to
% considerably more general settings (cardinality and other general
% constraints).
% Anthony et al~\cite{AGGN08} studied robust (or min-max) versions of
% $k$-median problems, giving logarithmic approximation algorithms.

The $k$-robust model was introduced in Feige et al.~\cite{FJMM07}, where they gave an $O(\log m \log n)$-approximation
for set cover; here $m$ and $n$ are the number of sets and elements in the set system.  To get such an
algorithm~\cite{FJMM07} first gave an $O(\log m \log n)$-approximation algorithm for $k$-max-min set-cover problem
using the online algorithm for set cover~\cite{AAABN03}.
% to get a greedy-style $O(\log  m \log n)$ approximation for this problem;
They then used the $k$-max-min problem as a separation oracle in an LP-rounding-based algorithm (\`a la~\cite{SS04}) to
get the same approximation guarantee for the $k$-robust problem. They also showed an $\Omega(\frac{\log m}{\log\log
m})$ hardness of approximation for $k$-max-min and $k$-robust set cover. Khandekar et al.~\cite{KKMS08} noted that the
LP-based techniques of~\cite{FJMM07} did not give good results for Steiner tree, and developed new combinatorial
constant-factor approximations for $k$-robust versions of Steiner tree, Steiner forest on trees and facility location.
Using our framework, the algorithm we get for Steiner tree can be viewed as a rephrasing of their algorithm---our proof
is arguably more transparent and results in a better bound. Our approach can also be used to get a slightly better
ratio than~\cite{KKMS08} for the Steiner forest problem on trees.

Constrained submodular maximization problems~\cite{NWF78I,NWF78II,S04,CCPV07,V08} appear very relevant at first sight:
e.g., the $k$-max-min version of min-cut (``find the $k$ sources whose separation from the sink costs the most'') is
precisely submodular maximization under a cardinality constraint, and hence is approximable to within $(1 - 1/e)$.  But
apart from min-cut, the other problems do not give us submodular functions to maximize, and massaging the functions to
make them submodular seems to lose logarithmic factors. E.g., one can use tree embeddings~\cite{FRT03} to reduce
Steiner tree to a problem on trees and make it submodular. In other cases, one can use online algorithms to get
submodular-like properties and obtain approximation algorithms for the $k$-max-min problems (as in~\cite{FJMM07}).
Though the LP-based framework~\cite{FJMM07} for $k$-robust problems does not seem to extend to problems other than set
cover, in the companion paper~\cite{GNR-rob-gen} we give a general algorithm for $k$-robust covering using offline and
online algorithms. However, since our goal in this paper is to obtain approximation factors better than the online
competitive ratios, it is unclear how to use these results.
% on submodular maximization.

% (in~\cite{GNR-rob-gen}).
%\lref[Appendix]{sec:gen-sets}).Eventually,

Considering the \emph{average} instead of the worst-case performance
gives rise to the well-studied model of stochastic
optimization~\cite{RS04, IKMM04}.  Some common generalizations of the
robust and stochastic models have been considered (see, e.g.,
Swamy~\cite{Swamy08} and Agrawal et al.~\cite{ADSY09}).
% \agnote{More here?}
% \agnote{Any others?}  \rnote{There is a Stanford paper by Yinyu Ye and
%   others that gives a model that trades of between stochastic and
%   robust optimization, using a transformation of the expectation
%   function from finance - we should cite this}

%Feige et al.~\cite{FJMM07} also considered the $k$-max-min set cover---they gave an $O(\log m \log n)$-approximation algorithm for this, and used it in the algorithm for $k$-robust set cover.

To the best of our knowledge, none of the $k$-max-min problems other than min-cut and set cover~\cite{FJMM07} have been
studied earlier.
%(which is submodular maximization)
The {\em $k$-\underline{min}-min} versions of covering problems (i.e. ``which $k$ demands are the {\em cheapest} to
cover?'') have been extensively studied for set cover~\cite{S97,GKS04}, Steiner tree~\cite{G05}, Steiner
forest~\cite{GHNR07}, min-cut and multicut~\cite{GNS06,R08}. However these problems seem to be related to the
$k$-max-min versions only in spirit.

%The $k$-max-min min-cut problem admits an $(1-1/e)$ approximation, since it is a case of cardinality constrained submodular maximization.
%Min-min versions of optimization problems have been widely studied:
% We note that the result of~\cite{FJMM07} for set-cover also holds in
% the more general setting, where each set has a different inflation
% factor.

\section{Notation and Definitions}
%\section{Notation}
\label{sec:notation}

\paragraph{Deterministic covering problems.}
A covering problem \cov has a ground-set $E$ of elements with costs
$c:E\rightarrow \mathbb{R}_+$, and $n$ covering requirements (often
called demands or clients), where the solutions to the $i$-th
requirement is specified---possibly implicitly---by a family
$\mathcal{R}_i \sse 2^E$ which is upwards closed (since this is a
covering problem).  Requirement $i$ is \emph{satisfied} by solution $S
\sse E$ iff $S\in \mathcal{R}_i$.  The covering problem $\cov = \langle
E,c, \{\mathcal{R}_i\}_{i=1}^n \rangle$ involves computing a solution
$S\sse E$ satisfying all $n$ requirements and having minimum cost
$\sum_{e\in S} c_e$.  E.g., in set cover, ``requirements'' are items to
be covered, and ``elements'' are sets to cover them with. In Steiner
tree, requirements are terminals to connect to the root and elements are
the edges; in multicut, requirements are terminal pairs to be separated,
and elements are edges to be cut.
% The input size is $\max\{|E|,m\}$.

\paragraph{Robust covering problems.}
This problem, denoted \rcov, is a {\em two-stage optimization} problem,
where elements are possibly bought in the first stage (at the given
cost) or the second stage (at cost $\lambda$ times higher). In the
second stage, some subset $\omega \sse[n]$ of requirements (also called
a \emph{scenario}) materializes, and the elements bought in both stages
must satisfy each requirement in $\omega$. Formally, the input to
problem \rcov consists of (a) the covering problem $\cov = \langle E,c,
\{\mathcal{R}_i\}_{i=1}^n\rangle$ as above, (b) a set $\Omega\sse
2^{[n]}$ of scenarios (possibly implicitly given), and (c)~an inflation
parameter $\lambda\ge 1$. A feasible solution to \rcov is a set of {\em
  first stage elements} $E_0\sse E$ (bought without knowledge of the
scenario), along with an {\em augmentation algorithm} that given any
$\omega\in \Omega$ outputs $E_\omega \sse E$ such that $E_0\cup
E_\omega$ satisfies all requirements in $\omega$.  The objective
function is to minimize:
%\vspace{-0.1in}
$c(E_0) + \lambda \cdot \max_{\omega\in\Omega} c(E_\omega)$.
Given such a solution, $c(E_0)$ is called the first-stage cost and
$\max_{\omega\in\Omega} c(E_\omega)$ is the second-stage cost.

\paragraph{$k$-robust problems.} In this paper,
%(except in \lref[Section]{sec:gen-sets}),
we deal with robust covering problems under {\em cardinality} uncertainty sets: i.e., $\Omega := \binom{[n]}{k} =
\{S\sse [n] \mid |S| = k\}$. We denote this problem by $\robkcov$.

\paragraph{Max-min problems.} Given a covering problem $\cov$ and a set
$\Omega$ of scenarios, the {\em max-min} problem involves finding a
scenario $\omega \in \Omega$ for which the cost of the min-cost solution
to $\omega$ is maximized. Note that by setting $\lambda=1$ in any robust
covering problem, \emph{the optimal value of the robust problem equals
  that of its corresponding max-min problem}. In a {\bf $k$-max-min
  problem} we have $\Omega = \binom{[n]}{k}$.

\subsection{The Abstract Properties we want from our Algorithms }
\label{sec:framework}

Our algorithms for robust and max-min versions of covering problems are
based on the following guarantee.
\begin{definition}\label{defn:algo}
  An algorithm is \emph{$(\alpha_1,\alpha_2,\beta)$-discriminating} iff
  given as input any instance of $\robkcov$ and a threshold $T$, the
  algorithm outputs
  \begin{inparaenum}[(i)]
  \item a set $\fst\sse E$, and
  \item an algorithm $\snd: \binom{[n]}{k} \rightarrow
    2^E$,
  \end{inparaenum}
  such that:
  \begin{OneLiners}
  \item[A.] For every scenario $D \in {[n] \choose k}$,
    \begin{OneLiners}
    \item[(i)] the elements in $\fst~\cup ~\snd(D)$ satisfy all
      requirements in $D$, and
    \item[(ii)] the resulting augmentation cost
      $c\left(\snd(D)\right)\le \beta\cdot T$.
    \end{OneLiners}
  \item[B.] Let $\Phistar$ and $\Tstar$ (respectively) denote the
    first-stage and second-stage cost of an optimal solution to the
    $\robkcov$ instance. If the threshold $T\ge \Tstar$ then the first stage
    cost $c(\fst)\le \alpha_1\cdot \Phistar + \alpha_2\cdot \Tstar$.
  \end{OneLiners}
\end{definition}

% Our approximation algorithms for both max-min and robust versions of
% covering problems are based on the following type of guarantee.
% \begin{definition}\label{defn:algo}
%   An algorithm for a robust covering problem \rcov is said to be
%   \emph{$(\alpha_1,\alpha_2,\beta)$-discriminating} iff given any
%   instance of \rcov, and a threshold $T$, the algorithm outputs a first
%   stage solution $\fst\sse E$ and an augmentation algorithm $\snd:\Omega
%   \rightarrow 2^E$ such that:
%   \begin{enumerate}
%   \item If $T\ge \Tstar$ then $c(\fst)\le \alpha_1\cdot \Phistar +
%     \alpha_2\cdot \Tstar$, where $\Phistar$ and $\Tstar$ (respectively)
%     denote the first-stage and second-stage cost of some optimal
%     solution to the \rcov instance.
%   \item For every $\omega\in \Omega$, the elements $\fst~\bigcup
%     ~\snd(\omega)$ satisfy all requirements in $\omega$.
%   \item For every $\omega\in \Omega$, the cost
%     $c\left(\snd(\omega)\right)\le \beta\cdot T$.
% \end{enumerate}
% \end{definition}

The next lemma shows why having a discriminating algorithm is sufficient to solve the
robust problem. % Loosely speaking, the real concern with the algorithm
% sketched in the introduction is that
The issue to address is that having guessed $T$ for the optimal second
stage cost, we have no direct way of verifying the correctness of that
guess---hence we choose the best among all possible values of $T$. For
$T \approx \Tstar$ the guarantees in \lref[Definition]{defn:algo} ensure
that we pay $\approx \Phistar + \Tstar$ in the first stage, and $\approx
\lambda \Tstar$ in the second stage; for guesses $T \ll \Tstar$, the
first-stage cost in guarantee~(2) is likely to be large compared to
$\opt$.

\begin{lemma}\label{lem:apx}
  If there is an $(\alpha_1,\alpha_2,\beta)$-discriminating algorithm
  for a robust covering problem $\robkcov$, then for every $\epsilon > 0$
  there is a $\left((1+\epsilon)\cdot\max\left\{\alpha_1,
    \beta+\frac{\alpha_2}\lambda\right\} \right)$-approximation algorithm for $\robkcov$.
\end{lemma}
\begin{proof} Let \a denote an algorithm for $\robkcov$ such that it is $(\alpha_1,\alpha_2,\beta)$-discriminating. Let
ground-set $E=[m]$, and $c_{max}:=\max_{e\in [m]} c_e$. By scaling, we may assume WLOG that all costs in the $\robkcov$
instance are integral. Let $\epsilon>0$ be any value as given by the lemma (where $\frac1\epsilon$ is polynomially
bounded), and $N:= \lceil \log_{1+\epsilon}~(m\,c_{max}) \rceil+1$; note that $N$ is polynomial in the input size.
Define $\calT:= \left\{
  \left(1+\epsilon\right)^i \mid 0\le i\le N\right\}$.

The approximation algorithm for $\robkcov$ runs the $(\alpha_1, \alpha_2, \beta)$-discriminating algorithm \a for every
choice of $T\in \calT$ (here $|\calT|$ is polynomially bounded), and returns the solution corresponding to:
$$\widetilde T := \arg\min \big\{c(\fst) + \lambda\cdot \beta\, T
\mid T\in \calT \big\}.
$$

Recall that \Tstar denotes the optimal second-stage cost, clearly $\Tstar\le m\cdot c_{max}$. Let $i^*\in\mathbb{Z}_+$
be chosen such that $(1+\epsilon)^{i^*-1} < \Tstar\le (1+\epsilon)^{i^*}$; also let $T':=(1+\epsilon)^{i^*}$ (note that
$T' \in\calT$). The objective value of the solution from \a for threshold $\widetilde T$  can be bounded as follows.
%is at most that for $T=\Tstar$, which
\begin{eqnarray*}
c\big(\Phi_{\widetilde T}\big) + \lambda\cdot \max_{\omega\in \Omega} ~c\big(\auga_{\widetilde T}(\omega)\big) &\le &
c(\Phi_{\widetilde T}) +
\lambda\cdot \beta\,\widetilde T \\
&\le  & c(\Phi_{T'}) + \lambda\cdot \beta\,T' \\
&\le & \left(\alpha_1\cdot \Phistar +\alpha_2\cdot \Tstar \right) + \left(1+\epsilon\right)\,\beta\lambda \cdot \Tstar \\
&\le  & \left(1+\epsilon\right)\cdot \left[ \alpha_1\cdot \Phistar +\left(\beta + \frac{\alpha_2}\lambda\right)\cdot
\lambda\Tstar\right].
\end{eqnarray*}
The first inequality follows from Property~A(ii) in \lref[Definition]{defn:algo}; the second by the choice of
$\widetilde T$; the third by Property~B (applied with threshold $T'\ge \Tstar$) in \lref[Definition]{defn:algo}, and
using $T'\le (1+\epsilon)\cdot \Tstar$. Thus this algorithm for $\robkcov$ outputs a solution that is a
$\left((1+\epsilon)\cdot\max\left\{\alpha_1,
    \beta+\frac{\alpha_2}\lambda\right\} \right)$-approximation.
\end{proof}

In the rest of the paper, we focus on providing discriminating
algorithms for suitable values of $\alpha_1,\alpha_2,\beta$.

\subsection{Additional Property Needed for $k$-max-min Approximations}
\label{sec:mm-framework}

As we noted above, a $k$-max-min problem is a $k$-robust problem where the inflation $\lambda = 1$ (which implies that
in an optimal solution $\Phistar = 0$, and $\Tstar$ is the $k$-max-min value). Hence a discriminating algorithm
immediately gives an approximation to the \emph{value}: for any $D \in \binom{[n]}{k}$, $\fst \cup \snd(D)$ satisfies
all demands in $D$, and for the right guess of $T \approx \Tstar$, the cost is at most $(\alpha_2 + \beta) \Tstar$. It
remains to output a bad $k$-set as well, and hence the following definition is useful.

\begin{definition}
  \label{def:strong-disc}
  An algorithm for a robust problem is \emph{strongly discriminating} if
  it satisfies the properties in Definition~\ref{defn:algo}, and when
  the inflation parameter is $\lambda = 1$ (and hence $\Phistar =
  0$), the algorithm also outputs a set $Q_T \in \binom{[n]}{k}$
  such that if $c(\fst) \geq \alpha_2 T$, the cost of optimally covering
  the set $Q_T$ is $\ge T$.
\end{definition}
Recall that for a covering problem $\Pi$, the cost of optimally covering the set of requirements $Q \in \binom{[n]}{k}$
is $\opt(Q) = \min \{c(E_Q) \mid E_Q \sse E \text{ and } E_Q \in \mathcal{R}_i~\forall i\in Q\}$.

\begin{lemma}
  \label{lem:apxstrong}
  If there is an
  $(\alpha_1,\alpha_2,\beta)$-\emph{strongly}-discriminating algorithm
  for a robust covering problem $\robkcov$, then for every $\epsilon > 0$
  there is an algorithm for $k$-max-min$(\Pi)$ that outputs a set $Q$
  such that for some $T$, the optimal cost of covering this set $Q$ is
  at least $T$, but every $k$-set can be covered with cost at most
  $(1+\epsilon)\cdot (\alpha_2 + \beta)\,T$.
\end{lemma}
\begin{proof} The approximation algorithm for $\mm(\Pi)$ is similar to that in \lref[Lemma]{lem:apx}. Let \a denote an algorithm for
the robust problem that is $(\alpha_1,\alpha_2,\beta)$ {\em strongly discriminating}. Recall that the $k$-max-min
instance corresponds to the $\robkcov$ instance with $\lambda=1$, and hence we will run algorithm \a on this robust
instance. Also from \lref[Definition]{defn:algo}, $\Tstar$ denotes the optimal second-stage cost of $\robkcov$, and its
optimal fist-stage cost $\Phistar=0$ (since $\lambda=1$). Note that the optimal value of the $k$-max-min instance also
equals \Tstar.

Let ground-set $E=[m]$, and $c_{max}:=\max_{e\in [m]} c_e$. By scaling, we may assume WLOG that all costs in the
instance are integral. Let $\epsilon>0$ be any value as given by the lemma (where $\frac1\epsilon$ is polynomially
bounded), and $N:= \lceil \log_{1+\epsilon}~(m\,c_{max}) \rceil+1$; note that $N$ is polynomial in the input size.
Consider the integral powers of $(1+\epsilon)$,
$$
\calT:=\{t_i\}_{i=0}^N, \quad \mbox{ where } t_i = \left(1+\epsilon\right)^{i} \mbox{ for }i=0,1,\cdots,N.
$$

The approximation algorithm for $\mm(\Pi)$ runs the strongly discriminating algorithm \a for every choice of $T\in
\calT$, and let $p\in\{1,\cdots,N\}$ be the smallest index such that $c(\Phi(t_{p}))\le \alpha_2 \;t_p$. Observe that
there must exist such an index since for all $T\ge \Tstar$, we have $c(\Phi_T)\le \alpha_2\, T^*\le \alpha_2\, T$
(property~B in \lref[Definition]{defn:algo}, using $\Phistar=0$), and clearly $\Tstar\le m\cdot c_{max}\le t_N$. The
algorithm then outputs $Q(t_{p-1})$ as the max-min scenario. Below we prove that it achieves the claimed approximation.
We have for all $T\ge 0$,
$$T^* = \max\left\{  \opt(D) : D\in {[n]\choose k} \right\}  \le \max \left\{  c(\fst) + c(\snd(D)) : D\in {[n]\choose k} \right\}
\le c(\fst)+\beta\; T.$$ Above, the inequalities are by conditions~A(i) and A(ii) of \lref[Definition]{defn:algo}.
Setting $T=t_p$ here, and by choice of $p$,
%we obtain
$$T^*\le c(\Phi(t_p)) + \beta\, t_p\le (\alpha_2+\beta)\,t_p.$$
Hence $t_p$ is a $(\alpha_2+\beta)$-approximation to the max-min value $T^*$. Now applying the condition of
\lref[Definition]{def:strong-disc} with $T=t_{p-1}$, since $c(\Phi(t_{p-1}))\ge \alpha_2\, t_{p-1}$ (by choice of index
$p$), we obtain that the minimum cost to cover requirements $Q(t_{p-1})$ is at least:
$$t_{p-1}=\frac{t_p}{1+\epsilon} \ge \frac{T^*}{(1+\epsilon)\cdot
  (\alpha_2+\beta)},$$
which implies the desired approximation guarantee.
\end{proof}

%-------------------------------------------------------------------------------------------------

\section{$k$-Robust Set Cover}
\label{sec:set-cover} \vspace{-2mm}

\newcommand{\newc}{\widehat{c}}
\newcommand{\fprime}{\mathcal{F}'}
\newcommand{\F}{\mathcal{F}}

Consider the $k$-robust set cover problem where there is a set system $(U, \F)$ with a universe of $|U| = n$ elements,
and $m$ sets in $\F$ with each set $R \in \F$ costing $c_R$, an inflation parameter $\lambda$, and an integer $k$ such
that each of the sets $\smash{\binom{U}{k}}$ is a possible scenario for the second-stage. Given \lref[Lemma]{lem:apx},
it suffices to show a discriminating algorithm as defined in \lref[Definition]{defn:algo} for this problem. The
algorithm given below is easy: pick all elements which can only be covered by expensive sets, and cover them in the
first stage.

% on the cardinality of the realized demand-set.

\begin{algorithm}
  \caption{Algorithm for $k$-Robust Set Cover}
  \begin{algorithmic}[1]
    \STATE \textbf{input:} $k$-robust set-cover instance and threshold $T$.
    \STATE \textbf{let} $\beta \gets 36 \ln m$, and $S\leftarrow
    \left\{v\in U \mid \mbox{ min cost set covering $v$ has cost at
        least } \beta\cdot \frac{T}{k}\right\}$.

    \STATE \textbf{output} first stage solution \fst as the
    Greedy-Set-Cover($S$).
   \STATE {\bf define} $\snd(\{i\})$ as the min-cost set covering $i$,
   for $i\in U\setminus S$; and  $\snd(\{i\})=\emptyset$ for $i\in S$.
     \STATE \textbf{output} second stage solution \snd where
     $\snd(D):=\bigcup_{i\in D}\snd(\{i\})$ for all  $D \sse U$.
  \end{algorithmic}
\end{algorithm}

% We note that this is a natural analogue of the Golovin et al.~\cite{GGR06} algorithm (for min-cut with $k=1$).

\begin{claim}[Property~A for Set Cover]\label{cl:sc-2nd}
  For all $T\ge 0$ and scenario $D \in \binom{U}{k}$, the sets
  $\fst\bigcup \snd(D)$ cover elements in $D$, and have cost
  $c(\snd(D))\le \beta\, T$.
\end{claim}
\begin{proof}
  The elements in $D \cap S$ are covered by \fst; and by definition of
  \snd, each element $i \in D\setminus S$ is covered by set
  $\snd(\{i\})$.  Thus we have the first part of the claim. For the
  second part, note that by definition of $S$, the cost of $\snd(\{i\})$
  is at most $\beta\,T/k$ for all $i\in U$.
\end{proof}

Below $H_n:=\sum_{i=1}^n \frac1i \approx \ln n$; recall that the greedy algorithm for set cover is an
$H_n$-approximation where $n$ is the number of elements in the given instance.

\begin{theorem}[Property~B for Set Cover]
  \label{th:main-sc}
  Let $\Phistar$ denote the optimal first stage solution (and its cost),
  and $\Tstar$ the optimal second stage cost. Let $\beta = 36 \ln m$. If
  $T\ge \Tstar$ then $c(\fst)\le H_n \cdot \left(\Phistar + 12\cdot
    \Tstar\right)$.
\end{theorem}

\begin{proof}
  We claim that there is a {\em fractional} solution $\bar{x}$ for the
  set covering instance $S$ with small cost $O(\Phistar + \Tstar)$,
  whence rounding this to an integer solution implies the theorem. For a
  contradiction, assume not: let every fractional set cover be
  expensive, and hence there must be a dual solution of large value. We
  then \emph{round this dual solution} to get a dual solution to a
  sub-instance with only $k$ elements that costs $> \Phistar + \Tstar$,
  which is impossible (since using the optimal solution we can solve
  every instance on $k$ elements with that cost).

  To this end, let $S'\sse S$ denote the elements that are {\em not}
  covered by the optimal first stage $\Phistar$, and let $\fprime \sse
  \F$ denote the sets that contain at least one element from $S'$. By
  the choice of $S$, all sets in $\fprime$ cost at least $\beta\cdot
  \frac{T}k\ge\beta\cdot \frac{\Tstar}k$. Define the ``coarse'' cost for
  a set $R \in \fprime$ to be $\newc_R = \ceil{\frac{c_R}{6\Tstar/k}}$.
  For each set $R\in\fprime$, since $c_R\ge \frac{\beta\Tstar}k \ge
  \frac{6\Tstar}k$, it follows that $\newc_R \cdot \frac{6\Tstar}k \in
  [c_R, 2\cdot c_R)$, and also that $\newc_R \ge \beta/6$.

Now consider the following primal-dual pair of LPs for the set cover instance
  with elements $S'$ and sets $\fprime$ having the coarse costs $\newc$.
$$
\begin{array}{llllll}
\min \,\, \sum_{R\in \fprime} \newc_R \cdot x_R & & \qquad  \qquad  \qquad  \max \,\, \sum_{e\in S'} y_e & \\
\sum_{R\ni e} x_R \ge 1,& \forall e\in S', & \qquad  \qquad \qquad  \sum_{e\in R} y_e \le \newc_R, & \forall R\in \fprime, \\
x_R \ge 0, & \forall R\in \fprime. & \qquad \qquad  \qquad  y_e\ge 0, & \forall e\in S'.
\end{array}
$$

  Let
  $\{x_R\}_{R \in \fprime}$ be an optimal primal and $\{y_e\}_{e \in S'}$ an
  optimal dual solution. The following claim bounds the (coarse) cost of
  these fractional solutions.

  \begin{claim}
    \label{cl:main}
    If $\beta = 36 \ln m$, then the LP cost is $\sum_{R\in\fprime}
    \newc_R\cdot x_R = \sum_{e\in S'} y_e \le 2\cdot k$.
  \end{claim}

  Before we prove \lref[Claim]{cl:main}, let us assume it and complete
  the proof of \lref[Theorem]{th:main-sc}. Given the primal LP solution
  $\{x_R\}_{R \in \fprime}$ to cover elements in $S'$, define an LP
  solution to cover elements in $S$ as follows: define $z_R = 1$ if
  $R\in\Phistar$, ${z}_R=x_R$ if $R\in \fprime \setminus \Phistar$; and
  ${z}_R=0$ otherwise.  Since the solution $\bar{z}$ contains $\Phistar$
  integrally, it covers elements $S \setminus S'$ (i.e. the portion of
  $S$ covered by $\Phistar$); since $z_R \geq x_R$, $\bar{z}$
  fractionally covers $S'$.  Finally, the cost of this solution is
  $\sum_R c_R z_R \leq \Phistar + \sum_R c_R x_R \le \Phistar+
  \frac{6\Tstar}k\cdot \sum_R \newc_R x_R$. But \lref[Claim]{cl:main}
  bounds this by $\Phistar+12\cdot \Tstar$. Since we have a LP solution
  of value $\Phistar + 12\Tstar$, and the greedy algorithm is an
  $H_n$-approximation relative to the LP value for set cover, this
  completes the proof.
\end{proof}

\lref[Claim]{cl:sc-2nd} and \lref[Theorem]{th:main-sc} show our algorithm for set cover to be an $(H_n, 12H_n, 36 \ln
m)$-discriminating algorithm. Applying \lref[Lemma]{lem:apx} converts this discriminating algorithm to an algorithm for
$k$-robust set cover, and gives the following improvement to the result of~\cite{FJMM07}.
\begin{theorem}
  There is an $O(\log m+\log n)$-approximation for $k$-robust set cover.
\end{theorem}

It remains to give the proof for \lref[Claim]{cl:main} above; indeed, that is where the technical heart of the result
lies.

\begin{proofof}{\lref[Claim]{cl:main}}
  Recall that we want to bound the optimal fractional set cover cost for
  the instance $(S', \fprime)$ with the coarse (integer) costs; $x_R$ and
  $y_e$ are the optimal primal and dual solutions. For a contradiction,
  assume that the LP cost $\sum_{R \in \fprime} \newc_R x_R = \sum_{e \in S'}
  y_e$ lies in the unit interval $((\gamma-1)k, \gamma k]$ for some integer
  $\gamma \ge 3$.

  Define integer-valued random variables $\{Y_e\}_{e \in S'}$ by
  setting, for each $e\in S'$ independently, $Y_e = \floor{y_e} + I_e$,
  where $I_e$ is a Bernoulli($y_e - \floor{y_e}$) random variable. We
  claim that whp the random variables $Y_e/3$ form a feasible dual---
  i.e., they satisfy all the constraints $\{\sum_{e \in R} (Y_e/3) \leq
  \newc_R\}_{R \in \fprime}$ with high probability.  Indeed, consider a
  dual constraint corresponding to $R\in \fprime$: since we have
  $\sum_{e \in R} \floor{y_e} \leq \newc_R$, we get that $\Pr[\sum_{e\in
    R} Y_e > 3\cdot \newc_R] \leq \Pr[\sum_{e\in R} I_e > 2\cdot
  \newc_R]$. But now we use a Chernoff bound~\cite{MR-book} to bound the
  probability that the sum of independent 0-1 r.v.s, $\sum_{e\in R} I_e$, exceeds twice its
  mean (here $\sum_{e\in R} E[I_e] \le \sum_{e\in R} y_e\le \newc_R$) by $e^{-\newc_R/3} \le e^{-\beta/18} \leq m^{-2}$, since each
  $\newc_R\ge \beta/6$ and $\beta = 36\cdot\ln m$. Finally, a trivial
  union bound implies that $Y_e/3$ satisfies all the $m$ contraints with
  probability at least $1-1/m$.  Moreover, the expected dual objective
  is $\sum_{e\in S'} y_e \ge (\gamma-1) k\ge 1$ (since $\gamma \geq 3$
  and $k \geq 1$), and by another Chernoff Bound, $\Pr[\sum_{e\in S'}
  Y_e > \frac{\gamma-1}2 \cdot k] \geq a$, where $a>0$ is some
  constant.  Putting it all together, with probability at least
  $a-\frac1m$, we have a {\em feasible} dual solution $Y'_e := Y_e/3$
  with objective value at least $\frac{\gamma-1}6\cdot k$.

  \emph{Why is this dual $Y'_e$ any better than the original dual
    $y_e$?} It is ``near-integral''---specifically, each $Y'_e$ is
  either zero or at least $\frac13$. So order the elements of $S'$ in
  decreasing order of their $Y'$-value, and let $Q$ be the set of the
  {\em first $k$ elements} in this order. The total dual value of
  elements in $Q$ is at least $\min\{ \frac{\gamma-1}{6}k, \frac{k}3\}
  \ge \frac{k}3$, since $\gamma \ge 3$, and each non-zero $Y'$ value is
  $\geq 1/3$. This valid dual for elements in $Q$ shows a lower bound of
  $\frac{k}3$ on minimum (fractional) $\newc$-cost to cover the $k$
  elements in $Q$. Using $c_R > \frac{3\Tstar}k\cdot \newc_R$ for each
  $R\in\fprime$, the minimum $c$-cost to fractionally cover $Q$ is $>
  \frac{3\Tstar}k\cdot \frac{k}3=\Tstar$.  Hence, if $Q$ is the realized
  scenario, the optimal second stage cost will be $> \Tstar$ (as no
  element in $Q$ is covered by $\Phistar$)---this contradicts the fact
  that OPT can cover $Q \in \smash{\binom{U}{k}}$ with cost at most
  $\Tstar$. Thus we must have $\gamma \le 2$, which completes the proof
  of \lref[Claim]{cl:main}.
\end{proofof}

\paragraph{The $k$-Max-Min Set Cover Problem.}
The proof of \lref[Claim]{cl:main} suggests how to get a $(H_n, 12H_n, 36 \ln m)$ strongly discriminating algorithm.
When $\lambda=1$ (and so $\Phistar = 0$), the proof shows that if $c(\fst) > 12 H_n \cdot T$, there is a randomized
algorithm that outputs $k$-set $Q$ with optimal covering cost $> T$ (witnessed by the dual solution having cost $>T$).
Now using \lref[Lemma]{lem:apxstrong}, we get the claimed $O(\log m + \log n)$ algorithm for the $k$-max-min set cover
problem. This nearly matches the hardness of $\Omega(\frac{\log
  m}{\log\log m} + \log n)$ given by~\cite{FJMM07}.

\noindent {\bf Remarks:} \ignore{We note that our $k$-robust algorithm also extends to the more general setting of
(uncapacitated) {\em Covering Integer Programs}; a CIP (see eg.~\cite{S99}) is given by $A\in [0,1]^{n\times m}$, $b\in
[1,\infty)^n$ and $c\in \mathbb{R}_+^m$, and the goal is to minimize $\{c^T\cdot x \mid Ax\ge b,\, x\in
\mathbb{Z}_+^m\}$.} The result above (as well as the~\cite{FJMM07} result) also hold in the presence of set-dependent
inflation factors---details appear in \lref[Appendix]{app:non-unif-sc}.  Results for the other covering problems do not
extend to the case of non-uniform inflation: this is usually inherent, and not just a flaw in our analysis. Eg.,
\cite{KKMS08} give an $\Omega(\log^{1/2-\epsilon} n)$ hardness for $k$-robust Steiner forest under just two distinct
inflation-factors, whereas we give an $O(1)$-approximation under uniform inflations (in
\lref[Section]{sec:steiner-forest}).

%%% Local Variables:
%%% mode: latex
%%% TeX-master: "k-rob"
%%% End:
%-------------------------------------------------------------------------------------------------
\section{$k$-Robust Minimum Cut}
\label{sec:minimum-cut}

We now consider the $k$-robust minimum cut problem, where we are given an undirected graph $G=(V,E)$ with edge
capacities $c:E\rightarrow \mathbb{R}_+$, a root $r\in V$, terminals $U\sse V$, inflation factor $\lambda$. Again, any
subset in $\smash{\binom{U}{k}}$ is a possible second-stage scenario, and again we seek to give a discriminating
algorithm. This algorithm, like for set cover, is non-adaptive: we just pick all the ``expensive'' terminals and cut
them in the first stage.

\begin{algorithm}
  \caption{Algorithm for $k$-Robust Min-Cut}
  \begin{algorithmic}[1]
    \STATE \textbf{input:} $k$-robust minimum-cut instance and threshold $T$.
    \STATE \textbf{let} $\beta \gets \Theta(1)$, and $S\leftarrow \{v\in
    U \mid \mbox{ min cut separating $v$ from root $r$ has cost at least
    } \beta\cdot \frac{T}{k}\}$.
    \STATE \textbf{output} first stage solution \fst as the minimum cut
    separating $S$ from $r$.
    \STATE {\bf define} $\snd(\{i\})$ as the min-$r$-$i$ cut in
    $G \setminus \fst$, for $i\in U\setminus S$; and
    $\snd(\{i\})=\emptyset$ for $i\in S$.
     \STATE \textbf{output} second stage solution \snd where
     $\snd(D):=\bigcup_{i\in D}\snd(\{i\})$ for all
     $D\sse U$.
  \end{algorithmic}
\end{algorithm}

\begin{claim}[Property~A for Min-Cut]\label{cl:mincut-2nd}
  For all $T \ge 0$ and $D \in \smash{\binom{U}{k}}$, the edges
  $\fst\bigcup \snd(D)$ separate the terminals $D$ from $r$; moreover,
  the cost $c(\snd(D))\le \beta\, T$.
\end{claim}
% \begin{proof}
%   Terminals in $D \cap S$ are separated from $r$ by \fst. By definition
%   of \snd, for each terminal $i\in D\setminus S$, the edges
%   $\snd(\{i\})$ form an $r-i$ cut. Thus we have the first part of the
%   claim. For the second part, note that by definition of $S$, the cost
%   of $\snd(\{i\})$ is at most $\beta\,T/k$ for all $i\in U$, and $|D|
%   \leq k$.
% \end{proof}

\begin{theorem}[Property~B for Min-Cut]
  \label{thm:mincut-main}
  Let $\Phistar$ denote the optimal first stage solution (and its cost),
  and $\Tstar$ the optimal second stage cost.  If $\beta\ge
  \frac{10e}{e-1}$ and $T\ge\Tstar$ then $c(\fst)\le 3\cdot \Phistar +
  \frac\beta{2}\cdot \Tstar$.
\end{theorem}

Here's the intuition for this theorem: As in the set cover proof, we claim that if the optimal cost of separating $S$
from the root $r$ is high, then there must be a dual solution (which prescribes flows from vertices in $S$ to $r$) of
large value. We again ``round'' this dual solution by aggregating these flows to get a set of $k$ terminals that have a
large combined flow (of value $> \Phistar + \Tstar$) to the root---but this is impossible, since the optimal solution
promises us a cut of at most $\Phistar + \Tstar$ for any set of $k$ terminals.

However, more work is required. For set-cover, each element was either covered by the first-stage, or it was not; for
cut problems, things are not so cut-and-dried, since both stages may help in severing a terminal from the root! So we
divide $S$ into two parts differently: the first part contains those nodes whose min-cut in $G$ is large (since they
belonged to $S$) but it fell by a constant factor in the graph $G \setminus \Phistar$. These we call ``low'' nodes, and
we use a Gomory-Hu tree based analysis to show that all low nodes can be completely separated from $r$ by paying only
$O(\Phistar)$ more (this we show in \lref[Claim]{clm:cut-low}). The remaining ``high'' nodes continue to have a large
min-cut in $G \setminus \Phistar$, and for these we use the dual rounding idea sketched above to show a min-cut of
$O(\Tstar)$ (this is proved in \lref[Claim]{clm:cut-high}). Together these claims imply
\lref[Theorem]{thm:mincut-main}.

To begin the proof of \lref[Theorem]{thm:mincut-main}, let $H:=G\setminus \Phistar$, and let $S_h\sse S$ denote the
``high'' vertices whose min-cut from the root in $H$ is at least $M:= \frac\beta2\cdot \frac{\Tstar}k$. The following
claim is essentially from Golovin et al.~\cite{GGR06}.
\begin{claim}[Cutting Low Nodes]
  \label{clm:cut-low}
  If $T\ge \Tstar$, the minimum cut in $H$ separating $S\setminus S_h$
  from $r$ costs at most $2\cdot \Phistar$.
\end{claim}
\begin{proof}
  Let $S':=S\setminus S_h$, and $t:=\beta\cdot \frac{\Tstar}k$. For
  every $v\in S'$, the minimum $r-v$ cut is at least $\beta\cdot
  \frac{T}k\ge \beta\cdot \frac{\Tstar}k=2M$ in $G$, and at most $M$ in
  $H$. Consider the {\em Gomory-Hu} (cut-equivalent) tree $\tf(H)$ on
  graph $H$ rooted at~$r$~\cite[Chap.~15]{Schr-book}.  For each $u \in S'$ let
  $D_u\sse V$ denote the minimum $r-u$ cut
%that is closest to the root
  in $\tf(H)$ where $u\in D_u$ and $r\not\in D_u$. Pick a subset
  $S''\sse S'$ of terminals such that the union of their respective
  min-cuts in $\tf(H)$ separate all of $S'$ from the root and their
  corresponding sets $D_u$ are disjoint (the set of cuts in tree $\tf(H)$ closest to the
  root $r$ gives such a collection).  It follows that (a) $\{D_u\mid u\in
  S''\}$ are disjoint, and (b) $F:=\cup_{u\in S''}
  \partial_H (D_u)$ is a feasible cut in $H$ separating $S'$ from $r$.
  Note that for all $u\in S''$, we have $c(\partial_H(D_u))\le M$ (since
  it is a minimum $r$-$u$ cut in $H$), and $c(\partial_G(D_u))\ge 2 M$
  (it is a feasible $r$-$u$ cut in $G$).  Thus $c(\partial_H(D_u))\le
  c(\partial_G(D_u)) - c(\partial_H(D_u)) =
  c(\partial_{\Phistar}(D_u))$. Now,
  $\ts c(F) \le \sum_{u\in S''} c(\partial_H(D_u))\le \sum_{u\in S''}
  c(\partial_{\Phistar}(D_u)) \le 2\cdot \Phistar$.
  The last inequality uses disjointness of $\{D_u\}_{u\in S''}$.  Thus
  the minimum $r-S'$ cut in $H$ is at most $2\Phistar$.
\end{proof}

\begin{claim}[Cutting High Nodes]
  \label{clm:cut-high}
  If $T \geq \Tstar$, the minimum $r$-$S_h$ cut in $H$ costs at most
  $\frac\beta{2}\cdot \Tstar$, when $\beta\ge \frac{10\cdot e}{e-1}$.
\end{claim}

\begin{proof}
  Consider an $r$-$S_h$ max-flow in the graph $H = G \setminus \Phistar$,
  and suppose it sends $\alpha_i\cdot M$ flow to vertex $i\in S_h$. By
  making copies of terminals, we can assume each $\alpha_i\in (0,1]$;
  the $k$-robust min-cut problem remains unchanged under making copies.
  Hence if we show that $\sum_{i\in S_h} \alpha_i \le k$, the total flow
  (which equals the min $r$-$S_h$ cut) would be at most $k\cdot M =
  \frac{\beta}{2}\cdot \Tstar$, which would prove the claim.  For a
  contradiction, we suppose that $\sum_{i\in S_h} \alpha_i >k$. We will
  now claim that there exists a subset $W \sse S_h$ with $|W|\le k$ such
  that the min $r$-$W$ cut is more than $\Tstar$, contradicting the fact
  that every $k$-set in $H$ can be separated from $r$ by a cut of value
  at most $\Tstar$. To find this set $W$, the following redistribution
  lemma (proved at the end of this theorem) is useful.

  \begin{lemma}[Redistribution Lemma]
    \label{lem:redistr}
    Let $N = (V,E)$ be a capacitated undirected graph. Let $X \sse V$ be
    a set of terminals such min-cut$_N(i,j) \geq 1$ for all nodes $i,j \in X$.
    For each $i\in X$, we are given a value $\epsilon_i\in (0,1]$. Then
    for any integer $\ell \leq \sum_{i \in X} \epsilon_i$, there exists
    a subset $W \sse X$ with $|W| \leq \ell$ vertices, and a feasible
    flow $\flow$ in $N$ from $X$ to $W$ so that (i) the total
    $\flow$-flow into $W$ is at least $\frac{1-e^{-1}}4\cdot \ell$ and
    (ii)~the \flow-flow out of each $i \in X$ is at most $\epsilon_i/4$.
  \end{lemma}

  We apply this lemma to $H = G \setminus \Phistar$ with terminal set
  $S_h$, but with capacities scaled down by $M$. Since for any cut
  separating $x,y \in S_h$, the root $r$ lies on one side on this cut
  (say on $y$'s side), min-cut$_H(x,y) \geq M$---hence the scaled-down
  capacities satisfy the conditions of the lemma. Now set $\ell = k$,
  and $\epsilon_i:=\alpha_i$ for each terminal $i\in S_h$; by the
  assumption $\sum_{i\in S_h} \epsilon_i =\sum_{i\in S_h} \alpha_i \ge k
  = \ell$. Hence \lref[Lemma]{lem:redistr} finds a subset $W \sse S_h$
  with $k$ vertices, and a flow \flow in (unscaled) graph $H$ such that
  \flow sends a total of at least $\frac{1-1/e}4\cdot kM$ units into
  $W$, and at most $\frac{\alpha_i}4\cdot M$ units out of each $i \in
  S_h$. Also, there is a feasible flow $g$ in the network $H$ that
  simultaneously sends $\alpha_i \cdot M$ flow from the root to each $i
  \in S_h$, namely the max-flow from $r$ to $S_h$. Hence the flow
  $\frac{g + 4\flow}{5}$ is {\em feasible} in $H$, and sends at least
  $\frac45\cdot \frac{1-1/e}4\cdot kM = \frac{1-1/e}5\cdot kM $ units
  from $r$ into $W$.  Finally, if $\beta > \frac{10\cdot e}{e-1}$, we
  obtain that the min-cut in $H$ separating $W$ from $r$ is greater than
  $\Tstar$: since $|W| \leq k$, this is a contradiction to the
  assumption that any set with at most $k$ vertices can separated from
  the root in $H$ at cost at most $\Tstar$.
\end{proof}

From \lref[Claim]{cl:mincut-2nd} and \lref[Theorem]{thm:mincut-main}, we obtain a
$(3,\frac\beta{2},\beta)$-discriminating algorithm for $k$-robust minimum cut, when $\beta\ge \frac{10e}{e-1}$. We set
$\beta=\frac{10e}{e-1}$ and use \lref[Lemma]{lem:apx} to infer that the approximation ratio of this algorithm is
$\max\{ 3, \frac\beta{2\lambda}+\beta\} = \frac\beta{2\lambda}+\beta$. Since picking edges only in the second-stage is
a trivial $\lambda$-approximation, the better of the two gives an approximation of
$\min\{\frac\beta{2\lambda}+\beta,~\lambda\}<17$. Thus we have,
\begin{theorem}[Min-cut Theorem]
  There is a 17-approximation algorithm for $k$-robust minimum cut.
\end{theorem}

% \subsection{The Proof of the Redistribution Lemma}
% \label{sec:card-mc-redistr}

% To prove \lref[Lemma]{lem:redistr}, we add each vertex $i \in X$ to a set
% $W$ independently with probability $\epsilon_i\, \ell/(\sum_i
% \epsilon_i)$. We add a super-source $s$ to the vertices in $X$, with the
% edge $e_i = (s, i)$ for $i \in X$ having capacity $\epsilon_i$. We want
% the excess at the source to be $b_s = -|W|$, and the excess at any
% vertex in $j \in W$ to be $b_j = 1$; note that any capacity-respecting
% $b$-transshipment in the network $N$ from $s$ to $W$ will satisfy the
% requirements of the lemma.  (Note: we need to bound the probability that
% $|W| > \ell$.) It follows from Hoffman's circulation theorem (and Gale's
% theorem) that such a flow is feasible if and only if $c(\partial U) \geq
% b(U)$ for all $U \sse V$.

% Consider any enumeration of the vertices of $V$

%\subsection{Proving the Redistribution Lemma}
%\label{sec:card-mc-redistr}

It now remains to prove the redistribution lemma. At a high level, the proof shows that if we add each vertex $i \in X$
to a set $W$ independently with probability $\epsilon_i\, \ell/(\sum_i \epsilon_i)$, then this set $W$ will (almost)
satisfy the conditions of the lemma whp. A natural approach to prove this would be to invoke Gale/Hoffman-type
theorems~\cite[Chap.~11]{Schr-book}: e.g., it is necessary and sufficient to show that $c(\partial V') \geq
|\text{demand}(V') - \text{supply}(V')|$ for all $V' \sse V$ for this random choice $W$.  But we need to prove such
facts for \emph{all} subsets, and all we know about the network is that the min-cut between any pair of nodes in $X$ is
at least $1$! Also, such a general approach is likely to fail, since the redistribution lemma is false for directed
graphs (see remark at the end of this section) whereas the Gale-Hoffman theorems hold for digraphs. In our proof, we
use undirectedness to fractionally pack Steiner trees into the graph, on which we can do a randomized-rounding-based
analysis.
% , and
% (essentially) show Gale's theorem for the set $W$ on each of these
% subtrees.

\begin{proofof}{\lref[Lemma]{lem:redistr} (Redistribution Lemma)}
To begin, we assume w.l.o.g. that the bounds $\epsilon_i = 1/P$ for all $i \in X$ for some integer $P$. Indeed, let
$P\in \mathbb{N}$ be large enough so that $\hat{\epsilon}_i = \epsilon_i P$ is an integer for each $i\in X$. Add, for
each $i \in X$, a star with $\hat{\epsilon}_i-1$ leaves centered at the original vertex~$i$, set all these new vertices
to also be terminals, and let all new edges have unit capacity. Set the new $\epsilon$'s to be $1/P$ for all terminals.
To avoid excess notation, call this graph $N$ as well; note that the assumptions of the lemma continue to hold, and any
solution $W$ on this new graph can be mapped back to the original graph.

% Consider any integer $\ell\le \sum_{i\in X}\epsilon_i =|X'|\cdot \frac1P
% $; then there is a subset $W'\sse X'$ with $|W'|\le \ell$, and flow
% $\f'$ (in $N'$) such that (i) $\flow'$ sends $(1-e^{-1})\cdot \ell/4$
% units into $W'$ and (ii) $\flow'$ sends at most $\frac1{4P}$ units out
% of each $j\in X'$.  Let $W=\{i\in X\mid W' \mbox{ contains some copy of
% }i\}$; clearly $|W|\le |W'|\le \ell$. Note that $\flow'$ also
% corresponds to a flow \flow in $N$ such that (i) \flow sends
% $(1-e^{-1})\cdot \ell/4$ units into $W$ and (ii) $\flow$ sends at most
% $\frac{e_i}{4P}=\epsilon_i/4$ units out of each $i\in X$ (since $N'$ has
% $e_i$ copies of each $i\in X$).  I.e.  the lemma holds for graph $N$
% with bounds $\epsilon_i$'s.  Thus it suffices to prove the existence of
% $W'$ and $\flow'$, which is established in the next claim.
% %the lemma when the bounds are all equal.
% %, and there is a feasible $r$-flow in $H'$ that sends $\frac1P$ flow to {\em each} terminal in $X'$.

% \begin{claim}
%   Let $W\sse X$ be an $\ell$-subset of $X$ chosen uniformly at random.
%   Then with constant probability, there is a flow $\flow$ in $N'$ such
%   that (i) $\flow'$ sends $\frac{1-e^{-1}}4\cdot \ell$ units into $W'$
%   and (ii) $\flow'$ sends at most $\frac1{4P}$ units out of each $j\in
%   X'$.
% \end{claim}

Let $c_e$ denote the edge capacities in $N$, and recall the assumption that every cut in $N$ separating $X$ has
capacity at least one. Since the natural LP relaxation for Steiner-tree has integrality gap of~$2$, this implies the
existence of Steiner trees $\{T_a\}_{a \in A}$ on the terminal set $X$ that fractionally pack into the edge capacities
$\bar{c}$. I.e., there exist positive multipliers $\{\lambda_a\}_{a \in
  A}$ such that $\sum_a \lambda_a = \frac12$, and $\sum_a \lambda_a\cdot
\bar{\chi}(T_a) \le \bar{c}$, where $\bar{\chi}(T_a)$ is the characteristic vector of the tree $T_a$. Choose $W \sse X$
by taking $\ell$ samples uniformly at random (with replacement) from $X$. We will construct the flow $\flow$ from $X$
to $W$ as a sum of flows on these Steiner trees. In the following, let $q:=|X|$; note that $\ell \leq |X|\epsilon =
q/P$.

Consider any fixed tree $T_a$ in this collection, where we think of the edges as having unit capacities. We claim that
in expectation, $\Omega(\ell)$ units of flow can be feasibly routed from $X$ to $W$ in $T_a$ such that each terminal
supplies at most $\ell/q$. Indeed, let $\tau_a$ denote an oriented Euler tour corresponding to $T_a$. Since the tour
uses any tree edge twice, any feasible flow routed in $\tau_a$ (with unit-capacity edges) can be scaled by half to
obtain a feasible flow in $T_a$. We call a vertex $v \in X$ {\em $a$-close} if there is some $W$-vertex located at most
$q/\ell$ hops from $v$ on the (oriented) tour $\tau_a$.  Construct a flow $\flow_a$ on $\tau_a$ by sending $\ell/q$
flow from each $a$-close vertex $v\in X$ to its nearest $W$-vertex along $\tau_a$. By the definition of $a$-closeness,
the maximum number of flow paths in $\flow_a$ that traverse an edge on $\tau_a$ is $q/\ell$; since each flow path
carries $\ell/q$ flow, the flow on any edge in $\tau_a$ is at most one, and hence $\flow_a$ is always feasible.

For any vertex $v\in X$ and a tour $\tau_a$, the probability that $v$ is {\em not} $a$-close is at most
$(1-\frac{q/\ell}{q})^{\ell} \le e^{-1}$; hence $v\in X$ sends flow in $\flow_a$ with probability at least $1-e^{-1}$.
Thus the expected amount of flow sent in $\flow_a$ is at least $(1-e^{-1}) |X| \cdot (\ell/q) = (1-e^{-1})\cdot \ell$.
Now define the flow $\flow := \frac12 \sum_a \lambda_a\cdot \flow_a$ by combining all the flows along all the Steiner
trees. It is easily checked that this is a feasible flow in $N$ with probability one. Since $\sum_a \lambda_a=
\frac12$, the expected value of flow $\flow$ is at least $\frac{1-1/e}{4} \ell$. Finally the amount of flow in \flow
sent out of any terminal is at most $\frac14\cdot \ell/q \le \frac1{4P}$. This completes the proof of the
redistribution lemma.
\end{proofof}

\medskip\textbf{The $k$-max-min Min-Cut Problem.} When $\lambda =
1$ and $\Phistar = 0$, the proof of \lref[Theorem]{thm:mincut-main} gives a randomized algorithm such that if the
minimum $r$-$S$ cut is greater than $\frac{\beta}2 T$, it finds a subset $W$ of at most $k$ terminals such that
separating $W$ from the root costs more than $T$ (witnessed by the dual value). Using this we get a randomized
$(3,\frac\beta{2},\beta)$ {\em strongly} discriminating algorithm, and hence a randomized $O(1)$-approximation
algorithm for $k$-max-min min cut from \lref[Lemma]{lem:apxstrong}. We note that for $k$-max-min min-cut, a
$(1-1/e)$-approximation algorithm was already known (even for directed graphs) via submodular maximization. However the
above approach has the advantage that it also extends to $k$-robust min-cut.

%\section{An Illustrative Example for $k$-Robust Min-Cut}\label{sec:bad-dir-eg}
\medskip\textbf{Bad Example for Directed Graphs.}
Let us show that our theorems for $k$-robust min-cut have to use the undirectedness of the graph crucially, and that
the theorems are in fact false for directed graphs. Consider the digraph $G$ with a root $r$, a ``center'' vertex $c$,
and $\ell$ terminals $v_1, v_2, \ldots, v_\ell$. This graph has arcs are $(c,r)$, $\{(r, v_i)\}_{i \in [\ell]}$ and
$\{(v_i, c)\}_{i \in [\ell]}$; each having unit capacity. Note that the min-cut between every $v_i$-$v_j$ pair is $1$,
but if we give each of the $v_i$'s $\epsilon_i = 1/\sqrt{\ell}$ flow, there is no way to choose $\sqrt{\ell}$ of these
vertices and collect a total of $\Omega(\sqrt{\ell})$ flow at these ``leaders''. This shows that the redistribution
lemma (\lref[Lemma]{lem:redistr}) is false for digraphs.

A similar example shows that that thresholded algorithms perform poorly for $k$-robust directed min-cut, even for
$k=1$. Consider graph $D$ with vertices $r$, $c$ and $\{v_i\}_{i \in [\ell]}$ as above. Graph $D$ has unit capacity
arcs $\{(v_i, r)\}_{i\in[\ell]}$, and $\sqrt{\ell}$ capacity arcs $(c,r)$ and $\{(v_i, c)\}_{i\in[\ell]}$. The
inflation factor is $\lambda = \sqrt{\ell}$. The optimal strategy is to delete the arc $(c,r)$ in the first stage.
Since $k=1$, one of the terminals $v_i$ demands to be separated from the root in the second stage, whence deleting the
edge $(v_i,r)$ costs $\lambda\cdot 1 = \sqrt{\ell}$ resulting in a total cost of $2\sqrt{\ell}$. However, any
threshold-based algorithm would either choose none of the terminals (resulting in a recourse cost of $\lambda
\sqrt{\ell} = \ell$), or all of them (resulting in a first-stage cost of at least $\ell$).

%-------------------------------------------------------------------------------------------------

\section{$k$-Robust Multicut}
\label{sec:robust-multicut}

We now consider the multicut problem: we are given an undirected graph $G=(V,E)$ with edge-costs $c:E\rightarrow
\mathbb{R}_+$, and $m$ vertex-pairs $\{s_i,t_i\}_{i=1}^m$. In the $k$-robust version, we are also given an inflation
parameter $\lambda$ and bound $k$ on the cardinality of the realized demand-set. Let $\Phistar$ denote the optimal
first stage solution (and its cost), and $\Tstar$ the optimal second stage cost; so $\opt=\Phistar+\lambda\cdot
\Tstar$. The algorithm (given below) is essentially the same as for minimum cut, however the analysis requires
different arguments.
%The exact value of $\beta$ is fixed later in the proof.

\begin{algorithm}
  \caption{Algorithm for $k$-Robust MultiCut}
  \begin{algorithmic}[1]
      \STATE \textbf{input:} $k$-robust multicut instance and threshold $T$.
    \STATE \textbf{let} $\rho :=O(\log n)$ be the approximation factor in R\"{a}cke's oblivious routing
    scheme~\cite{R08}, $\epsilon\in(0,\frac12)$ any constant, and $\beta := \rho\cdot
\frac{16\log n}{\epsilon \log\log n}$.
    \STATE \textbf{let} $S\leftarrow \{i\in [m] \mid \mbox{ min
      $s_i$-$t_i$ cut has cost at least } \beta\cdot \frac{T}{k}\}$.
    \STATE \textbf{output} first stage solution \fst as the $O(\log n)$-approximate multicut~\cite{GVY96} for $S$.
    \STATE {\bf define} $\snd(\{i\})$ as edges in the min $s_i-t_i$ cut, for $i\in [m]\setminus S$; and  $\snd(\{i\})=\emptyset$ for $i\in S$.
    \STATE \textbf{output} second stage solution \snd where $\snd(\omega):=\bigcup_{i\in \omega}\snd(\{i\})$ for all
     $\omega\sse [m]$.
  \end{algorithmic}
\end{algorithm}

\begin{claim}[Property~A for Multicut]\label{cl:multicut-2nd}
For all $T\ge 0$ and $\omega\sse [m]$, the edges $\fst\bigcup \snd(\omega)$ separate $s_i$ and $t_i$ for all
$i\in\omega$; additionally if $|\omega|\le k$ then the cost $c(\snd(\omega))\le \beta\, T$.
\end{claim}
\begin{proof}
Pairs in $\omega\cap S$ are separated by \fst. By definition of \snd, for each pair $i\in \omega\setminus S$ edges
$\snd(\{i\})$ form an $s_i-t_i$ cut. Thus we have the first part of the claim. For the second part, note that by
definition of $S$, the cost of $\snd(\{i\})$ is at most $\beta\,T/k$ for all $i\in [m]$.
\end{proof}

\begin{theorem}[Property~B for Multicut]
  \label{th:multicut-1st}
  If $T\ge \Tstar$ then $c(\fst)\le O(\log n)\cdot  \Phistar + O(\log^{2+\epsilon} n)\cdot \Tstar$.
\end{theorem}

To prove the theorem, the high level approach is similar to that for $k$-robust min-cut. We first show in
\lref[Lemma]{lem:multicut:GH} that the subset of pairs $\widetilde{S} \sse S$ whose min-cut fell substantially on
deleting the edges in $\Phistar$ can actually be completely separated by paying $O(1)\Phistar$. This is based on a
careful charging argument on the Gomory-Hu tree and generalizes the~\cite{GGR06} lemma from min-cut to multicut. Then
in \lref[Lemma]{lem:multicut-high} we show that the  remaining pairs in $S \setminus \widetilde{S}$ can be {\em
fractionally} separated at cost $O(\log^{1+\epsilon} n)\,\Tstar$. This uses the dual-rounding approach combined with
R\"{a}cke's oblivious routing scheme~\cite{R08}. Finally since the~\cite{GVY96} algorithm for multicut is relative to
the LP, this would imply \lref[Theorem]{th:multicut-1st}.

Let us begin by formally defining the cast of characters.  Let $H:=G\setminus \Phistar$ and
$M:=\beta\cdot\frac\Tstar{k}$. Define,
$$\ts \widetilde{S}:=\left\{i\in
  S\mid \mbox{ min cost $s_i$-$t_i$ cut in $H$ is less than
  }\frac{M}4\right\}$$ to be the set of pairs whose mincut in $G$ was at
least $M$, but has fallen to at most $M/4$ in $H = G \setminus \Phistar$.

\begin{lemma}
  \label{lem:multicut:GH}
  If $T\ge \Tstar$, there is a multicut separating pairs $\widetilde{S}$ in
  graph $H$ which has cost at most $2\,\Phistar$.
\end{lemma}
\begin{proof}
  We work with graph $H=(V,F)$ with edge-costs $c:F\rightarrow
  \mathbb{R}$. A {\em cluster} refers to any subset of vertices. A {\em
    cut equivalent tree} (c.f.~\cite{4Bill-book}), $P = (\n(P), E(P))$ is an
  edge-weighted tree on clusters $\n(P) = \{N_j\}_{j=1}^r$ such that:
  \begin{OneLiners}
  \item the clusters $\{N_j\}_{j=1}^r$ form a partition of $V$, and
  \item for any edge $e\in E(P)$, its weight in $P$ equals the $c$-cost
    of the cut corresponding to deleting this edge in $P$. I.e., if
    $(S_e, S_e^c)$ is the partition of $V$ obtained by unioning the
    vertices in the clusters belonging to the two connected components
    of $P \setminus \{e\}$, then $e$'s weight in $P$ equals $c(\delta(S_e)) = c(\delta(S_e^c))$.
  \end{OneLiners}

The \emph{Gomory-Hu} tree $P_{GH} = (V, E(P_{GH}))$ of $H$ is a
  cut-equivalent tree where the clusters are singleton vertices, and
  which has the additional property that for every $u,v\in V$ the
  minimum $u$-$v$ cut in $P_{GH}$ equals the minimum $u$-$v$ cut in
  $H$. For any cut-equivalent tree, a cluster $N\sse V$ is called
  {\em active} if there is some $i\in \widetilde{S}$ such that $|N\cap
  \{s_i,t_i\}|=1$; otherwise the cluster $N$ is called {\em dead}.
  We obtain a cut-equivalent tree $Q$ from $P_{GH}$ by repeatedly performing one of the following modifications:
  (1)~ for each edge having weight greater than $\frac{M}4$, merge the clusters corresponding to its end points; and
  (2)~ for each dead cluster, merge it with any of its neighboring clusters.
Note   that in the resulting tree $Q$, every edge in $E(Q)$ has weight at most  $\frac{M}4$, and every cluster in
$\n(Q)$ is active. Let $\mathcal{D}:=\bigcup_{N\in\n(Q)} \partial_H(N)$. In the next two claims we show that
$\mathcal{D}$ is a feasible multicut for $\widetilde S$ with cost at most $2\,\Phistar$.

\begin{Myquote}
\begin{claim}
   \label{cl:mcut-firststage-feas}
$\mathcal{D}$ is a feasible multicut separating pairs $\widetilde S$ in $H$.
\end{claim}
\begin{proof}
Clearly for each pair $i\in \widetilde S$, vertices $s_i$ and $t_i$ are in distinct active clusters of the Gomory-Hu
tree $P_{GH}$. Additionally there is some edge of weight less that $\frac{M}4$ on the $s_i-t_i$ path in $P_{GH}$: since
the minimum $s_i-t_i$ cut in $H$ is less than $\frac{M}4$. Observe that in obtaining tree $Q$ from $P_{GH}$, we never
contract two active clusters nor an edge of weight less that $\frac{M}4$. Thus $s_i$ and $t_i$ lie in distinct clusters
of $Q$. Since this holds for all $i\in \widetilde S$, the claim follows by definition of $\mathcal{D}$.
\end{proof}

\begin{claim}
   \label{cl:mcut-firststage-cost}
   The cost $c(\mathcal{D}) = \sum_{e\in \mathcal{D}} c_e \le 2\,\Phistar$, if $T\ge \Tstar$.
\end{claim}
\begin{proof}
Consider any cluster $N\in \n(Q)$. Since all clusters in $\n(Q)$ are active, $N$ contains exactly one of $\{s_i,t_i\}$
for some $i \in \widetilde S$. Hence the cut  $\partial_G(N)$ (in graph $G$) has cost at least $\beta\cdot \frac{T}k\ge
\beta\cdot \frac{\Tstar}k = M$, by definition of the set $S\supseteq \widetilde S$.

Let $\n_2(Q)\sse \n(Q)$ denote all clusters in $Q$ having degree at most two in $Q$. Note that $|\n_2(Q)|\ge \frac12
|\n(Q)|$. Using the above observation and the fact that clusters in $\n_2(Q)$ are disjoint, we have
\begin{equation}\label{eq:mcut-1st-gh} |\n_2(Q)|\, M\le \sum_{N\in\n_2(Q)} c(\partial_G(N)) = \sum_{N\in\n_2(Q)} \left( c(\partial_H(N)) +
c(\partial_{\Phistar}(N)) \right) \le \sum_{N\in\n_2(Q)}  c(\partial_H(N)) + 2 \Phistar.\end{equation}

We now claim that for any $N\in\n_2(Q)$, the cost $c(\partial_H(N))\le \frac{M}2$. Let $e_1$ and $e_2$ denote the two
edges incident to cluster $N$ in $Q$ (the case of a single  edge is easier). Let $(U_l,V\setminus U_l)$ denote the cut
corresponding to edge $e_l$ (for $l=1,2$) where $N\sse U_l$. Each of these cuts has cost $c(\partial_H(U_l))\le
\frac{M}4$ by property of cut-equivalent tree $Q$, and their union $\partial_H(U_1)\bigcup \partial_H(U_2)$ is the cut
separating $N$ from $V \setminus N$. Hence  it follows that   $c(\partial_H(N))\le 2\cdot\frac{M}4=\frac{M}2$. Using
this in~\eqref{eq:mcut-1st-gh} and simplifying, we obtain  $|\n(Q)|\, M\le 2\cdot |\n_2(Q)|\, M \le 8\,\Phistar$.

For each edge $e\in E(Q)$, let $D_e\sse F$ denote the edges in graph $H$ that go across the two components of
$Q\setminus \{e\}$. By the property of cut-equivalent tree $Q$, we have $c(D_e)\le \frac{M}4$. Since $\mathcal{D} =
\bigcup_{e\in E(Q)} D_e$,
$$c(\mathcal{D})\le \sum_{e\in E(Q)} c(D_e)\le |E(Q)|\, \frac{M}4\le |\n(Q)|\, \frac{M}4\le 2\Phistar$$
This proves the claim. \end{proof}
\end{Myquote}

Combining \lref[Claims]{cl:mcut-firststage-feas} and~\ref{cl:mcut-firststage-cost}, we obtain the lemma.
\end{proof}

Now we turn our attention to the remaining pairs $W:=S\setminus \widetilde{S}$, and show that there is a cheap cut
separating them in $H$. For this we use a dual-rounding argument, based on R\"{a}cke's oblivious routing scheme. Recall
that constant $0<\epsilon<\frac12$, $\rho =O(\log n)$ (R\"{a}cke's approximation factor), and $\beta = \rho\cdot
\frac{16\log n}{\epsilon \log\log n}$. Define $\alpha := e\rho\cdot \log^{\epsilon} n$.

%\agnote{Changes begin here}
\begin{lemma}
  \label{lem:multicut-high}
  There exists a \emph{fractional} multicut separating pairs $W$ in the
  graph $H$ which has cost $8\alpha\cdot\Tstar$.
\end{lemma}
\begin{proof}
For  any demand vector $d:W\rightarrow \mathbb{R}_+$, the optimal {\em
    congestion} of routing $d$ in $H$, denoted $\cg(d)$, is the smallest
  $\eta\ge 0$ such that there is a flow routing $d_i$ units of flow
  between $s_i$ and $t_i$ (for each $i\in W$), using capacity at most
  $\eta\cdot c_e$ on each edge $e\in H$.   Note that for every $i\in W$, the $s_i$-$t_i$ min-cut in $H$ has cost
  at least $L:= \frac{M}{4} = \frac\beta{4}\cdot \frac\Tstar{k}$. Hence for any $i\in W$, the optimal congestion
  for a \emph{unit} demand between $s_i$-$t_i$ (and zero between all
  other pairs) is at most $\frac1L$.

  Now consider R\"{a}cke's oblivious routing scheme~\cite{R08} as
  applied to graph $H$. This routing scheme, for each $i\in W$,
  prescribes a unit flow $\f_i$ between $s_i$-$t_i$ such that for every
  demand vector $d:W\rightarrow \mathbb{R}_+$,
  $$ \max_{e\in H} \frac{\sum_{i\in W} d_i\cdot \f_i(e)}{c_e}\le
  \rho\cdot \cg(d),\qquad \mbox{where $\rho=O(\log n)$};$$ i.e., the
  congestion achieved by using these oblivious templates to route the
  demand $d$ is at most $\rho$ times the best congestion possible for
  that particular demand $d$.

  Now consider a maximum multicommodity flow in $H$; suppose  that it sends $y_i\cdot
  \frac\Tstar{k}$ units between $s_i,t_i$ for each $i\in W$. For a
  contradiction, suppose that $\sum_{i\in W} y_i> 8\alpha \cdot k$.
  (Otherwise the maximum multicommodity flow, and hence its dual, the
  minimum fractional multicut is at most $8\alpha \Tstar$, and the lemma holds.) By making copies of vertex-pairs, we
  may assume that $y_i\in [0,1]$ for all $i\in W$; this does not change
  the $k$-robust multicut instance. Define a (not necessarily feasible)
  multicommodity flow $\g:= \sum_{i\in W} X_i\cdot \frac\Tstar{k}\cdot
  \f_i$, where each $X_i$ is an independent 0-1 random variable with
  $\Pr[X_i=1]=\frac{y_i}{\alpha}$, and $\f_i$ is the R\"{a}cke oblivious
  routing template. The flow has expected magnitude at least $\sum_i
  \frac{y_i}{\alpha} \frac\Tstar{k} \geq 8\Tstar$, and is the sum of
  $\{0, \frac{\Tstar}{k}\}$-valued random variables, hence by a Chernoff
  bound:

  \begin{Myquote}
  \begin{claim}
    With constant probability, the magnitude of flow $\g$ is at least $\Tstar$.
  \end{claim}

  \begin{claim}
    \label{cl:multicut-rand}
    The flow $\g$ is feasible with
    probability $1-o(1)$.
  \end{claim}

  \begin{proof}
    Fix any edge $e\in H$, and let $u_i(e) := \frac\Tstar{k}\cdot
    \f_i(e)$ for all $i\in W$. Note that the random process gives us a
    flow of $\sum_i X_i \cdot u_i(e)$ on the edge $e$.  The feasibility of the
    maximum multicommodity flow says that $\cg(\{y_i\cdot
    \frac\Tstar{k}\}_{i\in W})\le 1$. Since oblivious routing loses only a
    $\rho$ factor in the congestion,
    $$\ts \sum_{i\in W} y_i\cdot u_i(e) = \sum_{i\in W} y_i\cdot
    \frac\Tstar{k}\cdot \f_i(e) \le \rho\cdot c_e;$$ and the expected
    flow on edge $e$ sent by the random process above is $\sum_{i \in W}
    \frac{y_i}{\alpha}\cdot u_i(e) \leq \frac\rho\alpha  c_e$.

    Now, since the min $s_i$-$t_i$-cut is at least $L$ for any $i\in W$,
    a unit of flow can (non-obliviously) be sent between $s_i,t_i$ at
    congestion at most $\frac1L$. Hence using the oblivious routing
    template $\f_i$ incurs a congestion at most $\frac{\rho}L$. Hence,
    $$\ts u_i(e) = \frac\Tstar{k}\cdot \f_i(e) \le  \frac\Tstar{k}\cdot
    \frac{\rho}L\cdot c_e = \frac{4\rho}{\beta}\cdot c_e$$

%Note that this    implies that flipping one of the $X_i$s changes the flow on an edge    by at most a fraction $\frac{4\rho}\beta$ of the capacity, which by the
%    assumptions of the claim, is at most $\smash{\frac{1}{18 \ln n}} c_e$.

    We divide the individual contributions by the edge
    capacity and further scale up by $\frac\beta{4\rho}$ by defining new $[0,1]$-random
    variables $Y_i = \frac{X_i \cdot u_i(e)}{c_e} \cdot \frac\beta{4\rho}$. We get
    that $\mu:= E[\sum Y_i] \leq \frac{\beta}{4\alpha}$. Recall
    the Chernoff bound that says that for independent $[0,1]$-valued
    random variables $Y_i$,
    $$\Pr\left[\sum Y_i \geq (1+\delta)\cdot \mu\right]
    \leq \left(\frac{e}{1+\delta}\right)^{\mu(1+\delta)}$$
Using this with $\mu(1+\delta) = \frac{\beta}{4\rho}$ (hence $\delta+1\ge \frac{\alpha}{\rho}$) we
    get that
\[ \Pr\left[\sum_i X_i \cdot u_i(e) \geq c_e\right] = \Pr\left[\sum_i Y_i \geq
    \frac{\beta}{4\rho}\right] \leq \left( \frac{e\rho}{\alpha}\right)^{\beta/4\rho} = \exp\left(-\epsilon \log\log n\cdot \frac{4\log n}{\epsilon \log\log n}\right) = \frac{1}{n^4}, \]
since $\alpha = e\rho\cdot \log^\epsilon n$ and $\beta=\rho\cdot \frac{16\log n}{\epsilon \log\log n}$.   Now a trivial
union   bound over all $n^2$ edges gives the claim.
  \end{proof}
\end{Myquote}

  By another union bound, it follows that there exists a feasible
  multicommodity flow $\g$ that sends either zero or $\frac\Tstar{k}$
  units for each pair $i\in W$, and the total value of $\g$ is at least
  $\Tstar$. Hence there exists some $k$-set $W'\sse W$ such that the
  maximum multicommodity flow for $W'$ on $H$ is at least $\Tstar$. This
  contradicts the fact that every $k$-set has a multicut of cost less
  than $\Tstar$ in $H$. Thus we must have $\sum_{i\in W} y_i\le
  8\alpha \cdot k$, which implies \lref[Lemma]{lem:multicut-high}.
\end{proof}

Combining \lref[Lemmas]{lem:multicut:GH} and~\ref{lem:multicut-high},  we obtain a {\em fractional} multicut for pairs
$S$ in graph $G$, having cost $O(1)\cdot \Phistar+ O(\log^{1+\epsilon} n)\cdot \Tstar$. Since the Garg et
al.~\cite{GVY96} algorithm for multicut is an $O(\log n)$-approximation relative to the LP, we obtain
\lref[Theorem]{th:multicut-1st}.

From \lref[Claim]{cl:multicut-2nd} and \lref[Theorem]{th:multicut-1st}, it follows that this algorithm is $O\left(\log
n, \log^{2+\epsilon}n,\beta\right)$-discriminating for $k$-robust multicut. Since $\beta=O(\log^2n/\log\log n)$, using
\lref[Lemma]{lem:apx}, we obtain an approximation ratio of:
$$\max\left\{\log n, \frac{\log^2n}{\log\log n} + \frac{\log^{2+\epsilon}n}{\lambda}\right\}.$$

This is an $O\left(\frac{\log^2n}{\log \log n}\right)$-approximation when $\lambda\ge \log^{2\epsilon}n$. On the other
hand, when $\lambda\le \log^{2\epsilon}n$, we can use the trivial algorithm of buying all edges in the second stage
(using the GVY algorithm~\cite{GVY96}); this implies an $O(\log^{1+2\epsilon}n)$-approximation. Since
$\log^{1+2\epsilon}n=o\left(\frac{\log^2n}{\log\log n}\right)$, we obtain:

\begin{theorem}
  There is an $O\left(\frac{\log^2n}{\log\log n}\right)$-approximation
  algorithm for $k$-robust multicut.
\end{theorem}

\textbf{The $k$-max-min Multicut Problem.} The above ideas also lead to a $\left(c_1\cdot\log n,\, c_2\cdot\log^2 n,\,
c_3\cdot\log^2 n\right)$ strongly discriminating algorithm for multicut, where $c_1,c_2,c_3$ are large enough
constants. The algorithm is exactly Algorithm~5 with parameter $\beta:= \Theta(\log n)\cdot \rho$ with an appropriate
constant factor; recall that $\rho=O(\log n)$ is the approximation ratio for oblivious routing~\cite{R08}.
\lref[Lemma]{lem:multicut-high} shows that this algorithm is $\left(c_1\cdot\log n, c_2\cdot\log^2 n, c_3\cdot\log^2
n\right)$ discriminating (the parameters are only slightly different and the analysis still applies). To establish the
property in \lref[Definition]{def:strong-disc}, consider the case $\lambda=1$ (i.e. $\Phistar=0$) and $c(\fst)\ge
(c_2\log^2 n)\cdot T$. Since the~\cite{GVY96} algorithm is $O(\log n)$-approximate relative to the LP, this implies a
feasible multicommodity flow on pairs $W$ (since $\Phistar=0$ we also have $W=S$) of value at least $(c_4\,\log n)\cdot
T$ for some constant $c_4$. Then the randomized rounding (with oblivious routing) can be used to produce a $k$-set
$W'\sse W$ and a feasible multicommodity flow on $W'$ of value at least $T$; by weak duality it follows that the
minimum multicut on $W'$ is at least $T$ and so \lref[Definition]{def:strong-disc} holds. Thus by
\lref[Lemma]{lem:apxstrong} we get a randomized $O(\log^2n)$-approximation algorithm for $k$-max-min multicut.

\medskip\textbf{All-or-Nothing Multicommodity Flow.} As a possible use of the oblivious routing and randomized-rounding
based approach, let us state a result for the all-or-nothing multicommodity flow problem studied by Chekuri et
al.~\cite{CKS04}: given a capacitated undirected graph $G = (V,E)$ and source-sink pairs $\{s_i, t_i\}$ with demands
$d_i$ such that the min-cut$(s_i, t_i) = \Omega(\log^2 n) d_i$, one can approximate the maximum throughput to within an
$O(\log n)$ factor without violating the edge-capacities, even with $d_{\max} \geq c_{\min}$---the results of Chekuri
et al.~\cite{CKS04, CKS05} violated the edge-capacities in this case by an additive $d_{\max}$. This capacity violation
in the previous all-or-nothing results is precisely the reason they can not be directly used in our analysis of
$k$-robust multicut.

%------------------------------------------------------------------------------
\medskip\textbf{Summarizing Properties from Dual Rounding.} The proofs for all problems
considered so far (set cover, minimum cut, multicut) used certain dual rounding arguments. We now summarize the
resulting properties in a self-contained form.
%We remark that such properties
\begin{theorem}
  \label{lem:DR-set-cover}
Consider any instance of set cover; let $B\in\mathbb{R}_+$ and  $k\in\mathbb{Z}_+$ be values such that
  \begin{OneLiners}
  \item the set of minimum cost covering any element costs at least $36\,\ln m\cdot \frac{B}k$, and
  \item the minimum cost of covering any $k$-subset of elements is at most $B$.
  \end{OneLiners}
  Then the minimum cost of covering {\bf all} elements is at most
  $O(\log n)\cdot B$.
\end{theorem}

\begin{theorem}
  \label{lem:DR-min-cut}
Consider any instance of minimum cut in an $n$-vertex undirected graph with root $r$ and terminals $X$; let
$B\in\mathbb{R}_+$ and $k\in\mathbb{Z}_+$ be values such that
  \begin{OneLiners}
  \item the minimum cut separating $r$ and $u$ costs at least $10\cdot \frac{B}k$, for each terminal $u\in X$.
  \item the minimum cut separating $r$ and $S$ costs at most $B$, for every $k$-set $S\in {X \choose k}$.
  \end{OneLiners}
Then the minimum cut separating $r$ and all terminals $X$ costs at most $O(1)\cdot B$.
\end{theorem}

\begin{theorem}
  \label{lem:DR-multicut}
Consider any instance of multicut in an $n$-vertex undirected graph with source-sink pairs $\{s_i,t_i\}_{i\in [m]}$;
let $B\in\mathbb{R}_+$ and $k\in\mathbb{Z}_+$ be values such that
  \begin{OneLiners}
  \item the minimum $s_i-t_i$ cut costs at least $c\cdot \log^2n\cdot \frac{B}k$, for each pair $i\in [m]$.
  \item the minimum multicut separating pairs in $P$ costs at most $B$, for every $k$-set $P\in {[m] \choose k}$.
  \end{OneLiners}
Then the minimum multicut separating all  pairs $[m]$ costs at most $O(\log^2 n)\cdot B$. Here $c$ is a universal
constant that is independent of the multicut instance.
\end{theorem}

Such properties rely crucially on the specific problem structure, and cannot hold for general covering problems---even
for the Steiner-tree cost function on a tree metric (which, in fact, is submodular). Consider a tree on vertices
$\{r,u\}\bigcup\{v_i\}_{i=1}^n$ with root $r$ and terminals $\{v_i\}_{i=1}^n$.  The edges set contains $(r,u)$ with
cost $k$, and for each $i\in [n]$ the edge $(u,v_i)$ with cost one. For parameter $B=2k$, the cost for connecting any
single terminal to the root is $k+1>\frac{B}2$, whereas the cost for connecting {\em any} $k$-set of terminals is
$2k=B$. If a theorem like the ones above held, we might have hoped the cost to connect all the $n$ terminals would be
$\tilde{O}(B)$; instead it is $n+k > \frac{n}{2k}\cdot B$. This is also the reason why the algorithms for Steiner tree
and Steiner forest (which appear in the next section) are slightly more involved, and their proofs rely on a
primal-dual argument instead of dual rounding.

\ignore{if the cost function is submodular. Consider cost function $f$ on groundset $[n]$ defined as $f(S) =
\min\left\{ \alpha\frac{B}{k}\cdot|S|, \,\, \alpha\frac{B}{k} + \frac{(k-\alpha) B}{k^2}\cdot |S|\right\}$ for any
$\alpha\in \left(1,\frac{k}2\right)$. Function $f$ is clearly submodular. We have (a) $f(\{i\})\ge \alpha\cdot
\frac{B}k$ for all $i\in[n]$, and (b) $f(P)\le B$ for all $P\in {[n]\choose k}$. But still $f([n])\ge \frac{n}{2k}\cdot
B$.}
%-------------------------------------------------------------------------------------------------
\section{$k$-Robust Steiner Forest}
\label{sec:steiner-forest} \vspace{-2mm}

In $k$-robust Steiner forest, we have a graph $G=(V,E)$ with edge costs $c:E\rightarrow \mathbb{R}_+$, and a set $U
\sse V \times V$ of potential terminal pairs; any set in $\smash{\binom{U}{k}}$ is a valid scenario in the second
stage. For a set of pairs $S \sse V \times V$, the graph $G / S$ is obtained by identifying each pair in $S$ together;
$d_{G / S}(\cdot, \cdot)$ is the distance in this ``shrunk'' graph. The algorithm is given below.  This algorithm is a
bit more involved than the previous ones, despite a similar general structure: we maintain a set of ``fake'' pairs
$S_f$ that may not belong to $U$ for this case. The following analysis shows a constant-factor guarantee. (Without
\lref[lines]{sf-algo1}-\ref{sf-algo2}, the algorithm is more natural, but for that we can currently only show an
$O(\log n)$-approximation; it seems that an $O(1)$-approximation for that version would imply an $O(\log
n)$-competitiveness for online greedy Steiner forest.)

\begin{algorithm}
  \caption{Algorithm for $k$-Robust Steiner Forest}
  \begin{algorithmic}[1]
    \STATE \textbf{input:} $k$-robust Steiner forest instance and threshold $T$.
    \STATE \textbf{let} $\beta \gets \Theta(1), \gamma \gets \Theta(1)$
    such that $\gamma \leq \beta/2$.

    \STATE \textbf{let} $S_r, S_f, W \gets \emptyset$

    \WHILE{there exists a pair $(s,t) \in U$ with $d_{G / (S_r \cup
        S_f)}(s,t) > \beta\cdot \frac{T}{k}$}

    \STATE \textbf{let} $S_r \gets S_r \cup \{(s,t)\}$

    \STATE \textbf{if} $d_G(s, w) < \gamma \cdot \frac{T}{k}$ for
    some $w \in W$ \textbf{then} $S_f \gets S_f \cup \{(s,w)\}$
    \textbf{else} $W \gets W \cup \{s\}$ \label{sf-algo1}

    \STATE \textbf{if} $d_G(t, w') < \gamma \cdot \frac{T}{k}$ for
    some $w' \in W$ \textbf{then} $S_f \gets S_f \cup \{(t,w')\}$
    \textbf{else} $W \gets W \cup \{t\}$ \label{sf-algo2}

    \ENDWHILE

    \STATE \textbf{output} first stage solution \fst to be the
    2-approximate Steiner forest~\cite{AKR95,GW95} on pairs $S_r$ along with shortest-paths connecting every pair in $S_f$.

   \STATE \textbf{define} $\snd(\{i\})$ to be the edges on the $s_i-t_i$
   shortest-path in $G/(S_r\cup S_f)$, for each pair $i\in U$.

   \STATE \textbf{output} second stage solution \snd where
   $\snd(S):=\bigcup_{i\in D} \snd(\{i\})$ for all $D \sse U$.
  \end{algorithmic}
\end{algorithm}

\begin{claim}[Property~A for Steiner forest]\label{cl:card-sf-2nd}
  For all $T\ge 0$ and $D \in \binom{U}{k}$, the edges $\fst\bigcup
  \snd(D)$ connect every pair in $D$, and have cost $c(\snd(D))\le
  \beta\, T$.
\end{claim}
\begin{proof}
  The first part is immediate from the definition of \snd and the fact
  that \fst connects every pair in $S_r\cup S_f$. The second part
  follows from the termination condition $d_{G / (S_r \cup
    S_f)}(s_i,t_i) \le \beta\cdot \frac{T}{k}$ for all pairs $i\in U$;
  this implies $c(\snd(D)) \le \sum_{i\in D} c(\snd(\{i\})) \le
  \sum_{i\in D} d_{G / (S_r \cup S_f)}(s_i,t_i) \le \frac{|D|}k\cdot
  \beta\, T$.
\end{proof}

%To bound the first-stage cost, let us prove a few lemmas.
\begin{lemma}
  \label{lem:witness}
  The optimal value of the Steiner forest on pairs $S_r$ is at least
  $|W| \times \frac{\gamma}{2} \frac{T}{k}$.
\end{lemma}
\begin{proof}
  Consider the primal (covering) and dual (packing) LPs corresponding to
  Steiner forest on $S_r$. Note that for each pair $i\in S_r$, the distance
  $d_G(s_i,t_i)\ge \beta\cdot \frac{T}k\ge 2\gamma\cdot \frac{T}k$; so
  any ball of radius $\frac\gamma{2} \cdot \frac{T}k$ around a vertex in
  $S_r$ may be used in the dual packing problem since it separates some pair in $S_r$.  Observe that $W$
  consists of only vertices from $S_r$, and each time we add a vertex to
  $W$, it is at least $\gamma T/k$ distant from any other vertex in $W$.
  Hence we can feasibly pack dual balls of radius $\frac{\gamma}2\cdot
  \frac{ T}k$ around each $W$-vertex. This is a feasible dual to
  the Steiner forest instance on $S_r$, of value $|W| \gamma/2 \cdot
  T/k$. The lemma now follows by weak duality.
\end{proof}

\begin{lemma}
  \label{lem:num-fake}
  The number of ``witnesses'' $|W|$ is at least the number of
  ``real'' pairs $|S_r|$, and $|S_r|$ is at least the number of ``fake'' pairs $|S_f|$.
\end{lemma}
\begin{proof}
  Partition the set $S_r$ as follows: $S_g$ are the pairs where both
  end-points are added to $W$, $S_o$ are the pairs where exactly one
  end-point is added to $W$, and $S_b$ are the pairs where neither
  end-point is added to $W$. It follows that $|S_r|=|S_g|+|S_o|+|S_b|$
  and $|W|=2\cdot |S_g|+|S_o|$.

  Consider an auxiliary graph $H = (W, E(W))$ on the vertex set $W$
  which is constructed incrementally:
  \begin{OneLiners}
  \item When a pair $(s,t)\in S_g$ is added, vertices $s,t$ are added to
    $W$, and edge $(s,t)$ is added to $E(W)$.
  \item Suppose a pair $(s,t)\in S_o$ is added, where $s$ is added to
    $W$, but $t$ is not because it is ``blocked'' by $w' \in W$.  In
    this case, vertex $s$ is added, and edge $(s,w')$ is added to
    $E(W)$.
  \item Suppose a pair $(s,t)\in S_b$ is added, where $s$ and $t$ are
    ``blocked'' by $w$ and $w'$ respectively. In this case, no vertex is
    added, but an edge $(w,w')$ is added to $E(W)$.
  \end{OneLiners}

  \begin{Myquote}
    \begin{claim}
      \label{clm:st-one}
      At any point in the algorithm if $x,y \in W$ lie in the same
      component of $H$ then $d_{G / (S_f \cup S_r)}(x,y) = 0$.
    \end{claim}

    \begin{proof}
      By induction on the algorithm, and the construction of the graph $H$.
      \begin{OneLiners}
      \item Suppose pair $(s,t)\in S_g$ is added, then the claim is
        immediate. $H$ has one new connected component $\{s,t\}$ and
        others are unchanged. Since $(s,t)\in S_r$, $d_{G / (S_f \cup
          S_r)}(s,t) = 0$ and the invariant holds.
      \item Suppose pair $(s,t)\in S_o$ is added, with $s$ added to $W$
        and $t$ blocked by $w'\in W$. In this case, the component of $H$
        containing $w'$ grows to also contain $s$; other components are
        unchanged. Furthermore $(t,w')$ is added to $S_f$ and $(s,t)$ to
        $S_r$, which implies $d_{G / (S_f \cup S_r)}(s,w') = 0$. So the
        invariant continues to hold.
      \item Suppose pair $(s,t)\in S_b$ is added, with $s$ and $t$
        blocked by $w,w'\in W$ respectively. In this case, the
        components containing $w$ and $w'$ get merged; others are
        unchanged. Also $(s,w),(t,w')$ are added to $S_f$ and $(s,t)$ to
        $S_r$; so $d_{G / (S_f \cup S_r)}(w,w') = 0$, and the invariant
        continues to hold.
      \end{OneLiners}
      Since these are the only three cases, this proves the claim.
    \end{proof}

    \begin{claim}
      \label{clm:st-two}
      The auxiliary graph $H$ does not contain a cycle when $\gamma \leq
      \beta/2$
    \end{claim}

    \begin{proof}
      For a contradiction, consider the first edge $(x,y)$ that when added to $H$
      by the process above creates a cycle.  Let $(s,t)$ be the pair
      that caused this edge to be added, and consider the situation just
      before $(s,t)$ is added to $S_r$.  Since $(x,y)$ causes a cycle,
      $x,y$ belong to the same component of $H$, and hence $d_{G / (S_f
        \cup S_r)}(x,y) = 0$ by the claim above. But since $x$ is either
      $s$ or its ``blocker'' $w$, and $y$ is either $t$ or its blocker
      $w'$, it follows that $d_{G / (S_f \cup S_r)}(s,t) < 2\gamma \cdot
      \frac{T}{k}\le \beta \cdot \frac{T}{k}$. But this
      contradicts the condition which would cause $(s,t)$ to be chosen into $S_r$
      by the algorithm.
    \end{proof}
  \end{Myquote}

  Now for some counting. Consider graph $H$ at the end of the algorithm:
  $W$ denotes its vertices, and $E$ its edges.  From the construction of
  $H$, we obtain $|W|=2\cdot |S_g|+|S_o|$ and
  $|E|=|S_g|+|S_o|+|S_b|=|S_r|$.  Since $H$ is acyclic, $|S_r|=|E|\le
  |W|-1$. Also note that $|S_f|=2\cdot |S_b|+|S_o|=2\cdot
  |S_r|-|W|<|S_r|$. Thus we have $|W|\ge |S_r|\ge |S_f|$ as required in the lemma.
% i.e. $|S_b|\le |S_g|-1$. Thus $|W|\ge |S_g|+|S_o|+|S_b| + 1 = |S_r|+1$.
\end{proof}

\begin{theorem}[Property~B for Steiner forest]
  \label{th:card-sf-1st}
  Let $\Phistar$ denote the optimal first stage solution (and its cost),
  and $\Tstar$ the optimal second stage cost. If $T\ge \Tstar$ then
  $c(\fst)\le \frac{4\gamma}{\gamma-2}\cdot (\Phistar+\Tstar)$.
\end{theorem}
\begin{proof}
  Let $|S_r| = \alpha k$. Using \lref[Lemma]{lem:num-fake},
  \lref[Lemma]{lem:witness} and the optimal solution,
  \begin{equation}
    \label{eq:2}
    \ts \frac\gamma{2} \cdot \alpha \cdot T \le
    |W| \cdot \frac{\gamma}{2} \frac{T}{k} \le OPT(S_r) \leq
    \Phistar + \left\lceil{\frac{|S_r|}{k}} \right\rceil \Tstar \le
    \Phistar+ \Tstar + \alpha\cdot \Tstar
    \le \Phistar +\Tstar +\alpha\,T
  \end{equation}
  Thus $\alpha \cdot T \le \frac2{\gamma-2}\cdot (\Phistar+\Tstar)$ and
  $OPT(S_r)\le \frac{\gamma}{\gamma-2}\cdot (\Phistar+\Tstar)$. So the
  2-approximate Steiner forest on $S_r$ has cost at most
  $\frac{2\gamma}{\gamma-2}\cdot (\Phistar+\Tstar)$. Note that the
  distance between each pair in $S_f$ is at most $\gamma\cdot
  \frac{T}k$; so the total length of shortest-paths in $S_f$ is at most
  $|S_f|\cdot \gamma\cdot \frac{T}k \le |S_r|\cdot \gamma\cdot
  \frac{T}k$ (again by \lref[Lemma]{lem:num-fake}). Thus the algorithm's
  first-stage cost is at most $\frac{2\gamma}{\gamma-2}\cdot
  (\Phistar+\Tstar) + \alpha\gamma\cdot T\le
  \frac{4\gamma}{\gamma-2}\cdot (\Phistar+\Tstar)$.
\end{proof}

\begin{theorem}[Steiner Forest Main Theorem]
\label{th:card-sf} There is a 10-approximation for $k$-robust Steiner forest.
\end{theorem}
\begin{proof}
  Using \lref[Claim]{cl:card-sf-2nd} and \lref[Theorem]{th:card-sf-1st},
  we obtain a $(\frac{4\gamma}{\gamma-2}, \frac{4\gamma}{\gamma-2},
  \beta)$-discriminating algorithm (\lref[Definition]{defn:algo}) for
  $k$-robust Steiner forest. Setting $\beta=2\gamma$ and
  $\gamma:=2+2\cdot (1-1/\lambda)$, \lref[Lemma]{lem:apx} implies an
  approximation ratio of $\max\{ \frac{4\gamma}{\gamma-2},~
  \frac{4\gamma/\lambda}{\gamma-2} + 2\gamma \} \le
  4+\frac4{1-1/\lambda}$.  Again the trivial algorithm that only buys
  edges in the second-stage achieves a $2\lambda$-approximation. Taking
  the better of the two, the approximation ratio is
  $\min\{2\lambda,~4+\frac4{1-1/\lambda}\} < 10$.
\end{proof}

\paragraph{The $k$-max-min Steiner Forest Problem.}
We now extend the $k$-robust Steiner forest algorithm to be $(\frac{4\gamma}{\gamma-2}, \frac{4\gamma}{\gamma-2},
2\gamma)$ {\em strongly} discriminating (when $\gamma=3$). As shown earlier, it is indeed discriminating. To show that
\lref[Definition]{def:strong-disc} holds, consider the proof of \lref[Theorem]{th:card-sf-1st} when $\lambda =1$ (so
$\Phistar = 0$) and suppose $c(\fst)\ge \frac{4\gamma}{\gamma-2} T\ge 2\gamma \,T$. The algorithm to output the $k$-set
$Q$ has two cases.
\begin{enumerate}
\item If the number of ``real'' pairs $|S_r|\le k$ then $Q := S_r$.  We
  have:
  $$c(\fst)\le 2\cdot OPT(S_r) + \frac{\gamma T}k\, |S_f|\le 2\cdot
  OPT(S_r) + \frac{\gamma T}k \, |S_r|\le 2\cdot OPT(S_r) + \gamma T.$$
  The first inequality is by definition of \fst and since distance
  between each pair in $S_f$ is at most $\gamma\cdot \frac{T}k$, the
  second inequality is by \lref[Lemma]{lem:num-fake}, and the last
  inequality uses $|S_r|\le k$. Since $c(\fst)\ge 2\gamma T$, it follows
  that $OPT(S_r) \ge \gamma T/2\ge T$.
% Since $\fst$ consists of a 2-approximate Steiner  forest on $S_r$
\item If $|S_r|>k$ then the number of ``witnesses'' $|W|\ge |S_r|>k$, by \lref[Lemma]{lem:num-fake}.
  Let $Q \sse S_r$ be {\em any} $k$-set of pairs such that for each
  $i\in Q$ at least one of $\{s_i,t_i\}$ is in $W$. By the construction
  of $S_r$, we can feasibly pack dual balls of radius $\frac\gamma{2}
  \frac{T}k$ around each $W$-vertex, and so $OPT(Q)\ge |Q|\cdot
  \frac\gamma{2} \frac{T}k =\frac\gamma{2} T\ge T$.
\end{enumerate}
Thus we obtain a constant-factor approximation algorithm for $k$-max-min Steiner forest.

%Also, in~\cite{KKMS08}, a $3$-approximation for Steiner forest problem was given when the input graph is itself a tree.
%This bound can be improved to $2.25$ via our general approach; we defer details to the full version.

%------------------------------------------------------------------------------
\section{Final Remarks}
In this paper, we presented a unified approach to directly solving
$k$-robust covering problems and $k$-max-min problems. The results for
all problems except multicut are fairly tight (and nearly match the
best-possible for the offline versions). It would be interesting to
obtain an $O(\log n)$-approximation for $k$-robust and $k$-max-min
multicut.

As mentioned earlier, approximating the {\em value} of any max-min
problem reduces to the corresponding robust problem, for {\em any}
uncertainty set. We show in the companion paper~\cite{GNR-rob-gen} that
there is also a relation in the reverse direction---for any covering
problem that admits good offline and online approximation algorithms, an
algorithm for the max-min problem implies one for the robust
version. This reduction can be used to give algorithms for robust
covering under matroid- and knapsack-type uncertainty
sets~\cite{GNR-rob-gen}.

\paragraph{Acknowledgments.} We thank Chandra Chekuri, Ravishankar
Krishnaswamy, Danny Segev, and Maxim Sviridenko for invaluable
discussions.

%{\small
\bibliography{../robust,robust,../../abbrev,../../my-papers,../../embedding}
\bibliographystyle{plain}
%}

%\newpage

\appendix

%\section{Proofs of Main Reduction Lemmas}

%\subsection{Proof of \lref[Lemma]{lem:apx}}\label{app:algo-ovw}

%------------------------------------------------------------------------------
\section{$k$-Robust Steiner Tree}
\label{sec:steiner-tree} \vspace{-2mm}

In the $k$-robust Steiner tree, we are given a graph $G=(V,E)$ with edge costs $c:E\rightarrow \mathbb{R}_+$, a root
vertex $r$, and a set $U \sse V$ of potential terminals. Any set of $k$ terminals from $U$---i.e., any set in
$\smash{\binom{U}{k}}$---is a valid scenario in the second stage. Let $d(\cdot, \cdot)$ be the shortest-path distance
according to the edge costs. For a set $S \sse V$ of terminals, define the distance $d(v,S) := \min_{w \in S} d(v,w)$.

By the results in \lref[Section]{sec:framework}, a discriminating algorithm for this problem immediately gives us an
algorithm for the robust version, and this is how we shall proceed. Here is our discriminating algorithm for $k$-robust
Steiner tree: it picks a $\beta T/k$-net $S$ of the terminals in $U$, and builds a MST on $S$ as the first stage.

\begin{algorithm}
  \caption{Algorithm for $k$-Robust Steiner Tree}
  \begin{algorithmic}[1]
   \STATE {\bf input:} instance of $k$-robust Steiner tree and threshold $T$.
    \STATE \textbf{let} $\beta \gets \Theta(1)$, $S \gets \{r\}$.
    \WHILE{there exists a terminal $v \in U$ with $d(v, S) > \beta\cdot
      \frac{T}{k}$}

    \STATE $S \gets S \cup \{v\}$

    \ENDWHILE

    \STATE \textbf{output} first-stage solution \fst to be a
    minimum spanning tree on $S$.

   \STATE \textbf{for each} $i\in U$, define $\snd(\{i\})$ to be the edges
   on a shortest-path from  $i$ to $S$.

   \STATE \textbf{output} second-stage solution \snd where $\snd(D)
   := \bigcup_{i\in D} \snd(\{i\})$ for all $D \sse U$.
  \end{algorithmic}
\end{algorithm}

To show that the algorithm is discriminating, we need to show the two properties in \lref[Definition]{defn:algo}. The
first property is almost immediate from the construction: since every point in $U \setminus S$ is close to some point
in the net $S$, this automatically ensures that the second stage recourse cost is small.

\begin{claim}[Property~A for Steiner Tree]\label{cl:card-st-2nd}
  For all $T\ge 0$ and $D \in \binom{U}{k}$, the edges $\fst\cup
  \snd(D)$ connect the terminals in $D$ to the root $r$, and have cost
  $c(\snd(D))\le \beta\, T$.
\end{claim}
\begin{proof}
  From the definition of the second-stage solution, $\snd(D)$ contains
  the edges on shortest paths from each $D$-vertex to the set $S$.
  Moreover, \fst is a minimum spanning tree on $S$ (which in turn
  contains the root $r$). Hence $\fst \cup \snd(D)$ connects $D$ to the
  root $r$.  To bound the cost, note that by the termination condition
  in the \textbf{while} loop, every terminal $i\in U$ satisfies
  $d(i,S)\le \beta \frac{T}k$. Thus,
  $$c(\snd(D))\le \sum_{i\in D} c(\snd(\{i\})) = \sum_{i\in D}
  d(i,S) \le \frac{|D|}k \cdot \beta\,T.$$ This completes the proof that
  the algorithm above satisfies Property~A.
\end{proof}

It now remains to show that the algorithm satisfies Property~B as well. Let us show this for a sub-optimal settings of
values; we will improve on these values subsequently. The proof is dual-based and shows that if the cost of the MST on
$S$ were large, then the optimal first stage solution cost $\Phistar$ must have been large as well!

\begin{theorem}[Property~B for Steiner tree]\label{cl:card-st-1st-easy}
  Let $\Phistar$ denote the optimal first stage solution (and its cost),
  and $\Tstar$ the optimal second stage cost. If $T > T^*$ then the
  first stage cost $c(\fst) \leq 2\cdot \frac{\beta}{\beta-2}\cdot
  (\Phistar + \Tstar)$.
\end{theorem}
\begin{proof}
  Suppose $|S| = \alpha k$. We can divide up $S$ into $\ceil{\alpha}$
  sets $S_1, S_2, \ldots, S_{\ceil{\alpha}}$ with at most $k$ terminals
  each, and let $E(S_i)$ denote the second-stage edges bought by the
  optimal solution under scenario $S_i$. Hence $\Phistar \cup (\cup_{i
    \leq \ceil{\alpha}} E(S_i))$ is a feasible solution to the Steiner
  tree on $S$ of cost at most $\Phistar + \ceil{\alpha}\cdot \Tstar$.
 % and the algorithm's first-stage cost would be at most $\rho_{ST}$
 % times that. % For small values of $\alpha$, we can bound the
%   first stage cost directly by $\rho_{ST} ( \Phistar +
%   \ceil{\alpha}\cdot \Tstar)$, and for large values of $\alpha$
  Also, since each of the points in $S$ is at least at distance $\beta
  T/k$ from each other, we get (below $OPT(S)$ is the length of the
  minimum Steiner tree on $S$),
  \begin{equation*}
    \ts \frac\beta{2} \cdot \alpha T = |S|\cdot\frac{\beta}{2}\cdot \frac{T}{k} \leq OPT(S) \leq
    \Phistar + \ceil{\alpha} \cdot \Tstar \le     \Phistar + (\alpha +1)
    \cdot \Tstar \le \Phistar +\Tstar + \alpha T.
  \end{equation*}
  Hence $\alpha T\le \frac{2}{\beta-2}(\Phistar+\Tstar)$ and $ OPT(S) \le \frac{\beta}{\beta-2}\cdot (\Phistar+\Tstar)$;
  since the MST heuristic is a $2$-approximation to the optimal Steiner  tree, we get the theorem.
\end{proof}

Combining \lref[Claim]{cl:card-st-2nd} and \lref[Theorem]{cl:card-st-1st-easy} shows that our algorithm is a
$(\frac{2\beta}{\beta-2}, \frac{2\beta}{\beta-2}, \beta)$-discriminating algorithm for $k$-robust Steiner tree.
Setting, say, $\beta = 4$ and applying \lref[Lemma]{lem:apx} gives us an $\max(4, 4 + 4) = 8$-approximation for
$k$-robust Steiner tree. In the next subsection, we will show how to improve this guarantee.

\subsection{Improved Approximation for Steiner Tree}
\label{sec:improve-steiner}

In the previous analysis, we just wanted to show the main ideas and hence were somewhat sloppy with the analysis. Let
us now show how to get a tighter bound using a fractional analysis.

\begin{theorem}[Improved Property~B for Steiner Tree]
  \label{th:card-st-1st}
  If $T\ge \Tstar$ then $c(\fst)\le \frac{2\beta}{\beta-2}\cdot \Phistar
  + 2\cdot \Tstar$.
\end{theorem}
\begin{proof}
  Firstly suppose $|S|\le k$: then it is clear that there is a Steiner
  tree on $\{r\}\cup S$ of cost at most $\Phistar+\Tstar$, and the
  algorithm finds one of cost at most twice that. In the following
  assume that $|S|>k$.

  Let $LP(S)$ denote the minimum length of a {\em fractional} Steiner
  tree on terminals $\{r\}\cup S$. Since each of the points in $S$ is at
  least at distance $\beta\cdot \frac{T}{k}\ge \beta\cdot
  \frac\Tstar{k}$ from each other, we get $LP(S)\ge \frac\beta{2k}\cdot
  |S|\cdot \Tstar$. We now construct a fractional Steiner tree
  $x:E\rightarrow \mathbb{R}_+$ of small length.
  % Firstly set $x_e=1$ for all $e\in\Phistar$.
  Number the terminals in $S$ arbitrarily, and for each $1\le j\le |S|$
  let $A_j = \{j,j+1,\cdots,j+k-1\}$ (modulo $|S|$). Let $\Pi_j\sse
  E\setminus \Phistar$ denote the second-stage edges bought in the
  optimal solution under scenario $A_j$: so $\Phistar\cup \Pi_j$ is a
  Steiner tree on terminals $\{r\}\cup A_j$, and $c(\Pi_j)\le \Tstar$.
  Define $x:=\chi(\Phistar) + \frac1k\cdot \sum_{j=1}^{|S|}
  \chi(\Pi_j)$. We claim that $x$ supports unit flow from $r$ to any
  $i\in S$: note that there are $k$ sets $A_{i-k+1},\cdots ,A_i$ that
  contain $i$, and for each $i-k+1\le j\le i$, we have $\frac1k\cdot
  \left(\chi(\Phistar)+\chi(\Pi_j)\right)$ supports $\frac1k$ flow from
  $r$ to $i$. Thus $x$ is a feasible fractional Steiner tree on
  $\{r\}\cup S$, of cost at most $\Phistar+\frac{|S|}k\cdot\Tstar$.
  Combined with the lower bound on $LP(S)$,
  \begin{equation}
    \label{eq:1}
    \ts |S|\cdot\frac{\beta}{2}\cdot \frac{\Tstar}{k} \leq LP(S) \leq
    \Phistar + \frac{|S|}k\cdot\Tstar.
  \end{equation}
  Thus we have $LP(S) \le \frac{\beta}{\beta-2}\cdot \Phistar$, which
  implies the theorem since the minimum spanning tree on $\{r\}\cup S$
  costs at most twice $LP(S)$.
\end{proof}

From \lref[Claim]{cl:card-st-2nd} and \lref[Theorem]{th:card-st-1st}, we now get that the algorithm is
$(\frac{2\beta}{\beta-2}, 2, \beta)$-discriminating. Thus, setting $\beta=2-\frac1\lambda +\sqrt{4+1/\lambda^2}$ and
applying \lref[Lemma]{lem:apx}, we get the following approximation ratio.
$$\max \left\{ \frac{2\beta}{\beta-2},~\frac2\lambda + \beta \right\} =
2+\frac1\lambda+\sqrt{4+\frac{1}{\lambda^2}}.$$ On the other hand, the trivial algorithm which does nothing in the
first stage is a $1.55\cdot \lambda$ approximation. Hence the better of these two ratios gives an approximation bound
better than $4.5$.

\medskip
{\bf The $k$-max-min Steiner Tree Problem.} We show that the above algorithm can be extended to be
$(\frac{2\beta}{\beta-2}, 2, \beta)$ strongly discriminating. As shown above, it is indeed discriminating. To show that
\lref[Definition]{def:strong-disc} holds, consider the proof of \lref[Theorem]{th:card-st-1st} when $\lambda =1$ (so
$\Phistar = 0$) and suppose that $c(\fst)\ge 2 \,T$. The algorithm to output the $k$-set $Q$ proceeds via two cases.
\begin{enumerate}
\item  If $|S|\le k$ then $Q := S$. The minimum Steiner tree on $Q$ is at least half its MST, i.e. at
least $\frac12 c(\fst) \ge T$.

\item If $|S|>k$ then $Q\sse S$ is {\em any} $k$-set; by the construction of $S$, we can feasibly pack dual balls of
radius $\beta \frac{T}k$ around each $Q$-vertex, and so $LP(Q)\ge \beta T\ge T$. Thus the minimum Steiner tree on $Q$
is at least $T$.
\end{enumerate}

\subsection{Unrooted Steiner tree}
\label{sec:unrooted-steiner}

We note that the $k$-robust Steiner tree problem studied above differs from~\cite{KKMS08} since there is no root in the
model of~\cite{KKMS08}. In the unrooted version, any subset of $k$ terminals appear in the second stage, and the goal
is to connect them {\em amongst each other}. We show that a small modification in the proof implies that Algorithm~1
(where $r\in U$ is set to an arbitrary terminal) achieves a good approximation in the unrooted case as well. This
algorithm is essentially same as the one used by~\cite{KKMS08}, but with different parameters: hence our framework can
be viewed as generalizing their algorithm. Our proof is somewhat shorter and gives a slightly better approximation
ratio.
%\agnote{How do the proofs compare? Do we need to say  ``Our proof is
%somewhat shorter but essentially the same: we give it
%  here for completeness''? }

Below, $\Phistar$ and $\Tstar$ denote the optimal first and second stage costs for the given unrooted instance.  It is
clear that \lref[Claim]{cl:card-st-2nd} continues to hold in this case as well: hence Property~A of
\lref[Definition]{defn:algo} is satisfied. We next bound the first stage cost of the algorithm (i.e. Property~B of
\lref[Definition]{defn:algo}).

\begin{theorem}[Property~B for Unrooted Steiner Tree]
  If $T\ge \Tstar$ then
  $c(\fst)\le \frac{2\beta}{\beta-2}\cdot \Phistar+2\Tstar$.
\end{theorem}
\begin{proof}
  Firstly suppose $|S|\le k$: then it is clear that there is a Steiner
  tree on $S$ of cost at most $\Phistar+\Tstar$, and the algorithm finds
  one of cost at most twice that. In the following assume that $|S|>k$.

  Let $LP(S)$ denote the minimum length of a {\em fractional} Steiner
  tree on terminals $S$ (recall, no root here). Since each of the points
  in $S$ is at least at distance $\beta\cdot \frac{T}{k}\ge\beta\cdot
  \frac\Tstar{k}$ from each other, we get $LP(S)\ge \frac\beta{2k}\cdot
  |S|\cdot \Tstar$. We now construct a fractional Steiner tree
  $x:E\rightarrow \mathbb{R}_+$ of small length.
  % Firstly set $x_e=1$ for all $e\in\Phistar$.
  Number the terminals in $S$ arbitrarily, and for each $1\le j\le |S|$
  let $A_j = \{j,j+1,\cdots,j+k-1\}$ (modulo $|S|$). Let $\Pi_j\sse
  E\setminus \Phistar$ denote the second-stage edges bought in the
  optimal solution under scenario $A_j$: so $\Phistar\cup \Pi_j$ is a
  Steiner tree on terminals $A_j$, and $c(\Pi_j)\le \Tstar$. Define
  $x:=\chi(\Phistar) + \frac1k\cdot \sum_{j=1}^{|S|} \chi(\Pi_j)$.
  \begin{Myquote}
    \begin{claim}
      For any $i\in S$, $x$ supports a unit flow from terminal $i$ to
      $i+1$ (modulo $|S|$).
    \end{claim}
    \begin{proof} Note that there are $k-1$ sets $A_{i-k+2},\cdots ,A_i$
      that contain both $i$ and $i+1$. Let $J:= \{i-k+2,\cdots,i\}$. So
      for each $j\in J$, we have $\frac1k\cdot
      \left(\chi(\Phistar)+\chi(\Pi_j)\right)$ supports $\frac1k$ flow
      from $i$ to $i+1$. Furthermore, $\left(\cup_{l\in S\setminus J}
        \Pi_l \right)\cup \Phistar$ is a Steiner tree connecting
      terminals $\cup_{l\in S\setminus J} A_l\supseteq \{i,i+1\}$; i.e.
      $\frac1k\cdot \left(\chi(\Phistar)+ \sum_{l\in S\setminus J}
        \chi(\Pi_l)\right)$ also supports $\frac1k$ flow from $i$ to
      $i+1$. Thus we obtain the claim.
    \end{proof}
  \end{Myquote}

  Thus $x$ is a feasible fractional Steiner tree on terminal $S$, of
  cost at most $\Phistar+\frac{|S|}k\cdot\Tstar$.  Combined with the
  lower bound on $LP(S)$,
  \begin{equation}
    \label{eq:1app}
    \ts |S|\cdot\frac{\beta}{2}\cdot \frac{\Tstar}{k} \leq LP(S) \leq
    \Phistar + \frac{|S|}k\cdot\Tstar.
  \end{equation}
  Thus we have $LP(S) \le \frac{\beta}{\beta-2}\cdot \Phistar$, which
  implies the theorem since the minimum spanning tree on $S$ costs at
  most twice $LP(S)$.
\end{proof}

Thus by the same calculation as in the rooted case, we obtain a result that slightly improves on the constants obtained
by~\cite{KKMS08} for the same problem.
\begin{theorem}\label{th:card-st}
  There is a 4.5-approximation algorithm for (unrooted) $k$-robust
  Steiner tree.
\end{theorem}

%------------------------------------------------------------------------------
\section{$k$-Robust Set Cover with Non-uniform Inflation}\label{app:non-unif-sc}
Consider the $k$-robust set cover problem where there is a set system $(U, \{R_j\}_{j=1}^m)$ with a universe of $n$
elements and $m$ sets with cost-vectors $b, c \in \R_+^m$ (for first and second stage resp.), and a bound $k$ on the
cardinality of the realized demand-set.
%Note that the uncertainty set here is $\Omega := \{\omega\sse U : |\omega|\le k\}$.
The model considered in \lref[Section]{sec:set-cover} is the special case when $c=\lambda\,b$ for some uniform inflation
factor $\lambda$. Here we consider the general case of {\em set-dependent inflation}, and show that the same result
holds. We may assume WLOG that the first-stage cost for each set is at most its second-stage cost, i.e. $b\le c$. (If
some set $R$ has $c_R<b_R$, then we pretend that its first-stage cost is $c_R$; and if $R$ is chosen into the
first-stage solution it can be always bought in the second stage).

%For a given instance, let $\Phistar$ denote the optimal first stage solution, and $\Tstar$ the optimal second stage $c$-cost; so the optimal value $\opt=b(\Phistar)+\Tstar$.
Under non-uniform inflations, the definition of an $(\alpha_1,\alpha_2,\beta)$-discriminating algorithm is the same as
\lref[Definition]{defn:algo} where Condition~B is replaced by:
\begin{OneLiners}
\item[B'.] Let $\Phistar$ denote the optimal first stage solution, and $\Tstar$ the optimal second stage
$c$-cost (hence the optimal value $\opt=b(\Phistar)+\Tstar$). If the threshold $T\ge \Tstar$ then the first stage cost
$b(\fst)\le \alpha_1\cdot \Phistar + \alpha_2\cdot \Tstar$.
\end{OneLiners}

It can be shown exactly as in \lref[Lemma]{lem:apx}, that any such algorithm  is a
$\max\{\alpha_1,~\alpha_2+\beta\}$-approximation for $k$-robust set cover. Note that the factor $\alpha_2$ was scaled
down by $\lambda$ in the uniform inflation case (\lref[Lemma]{lem:apx}). The algorithm and analysis here are very similar
to that for $k$-robust set-cover under uniform inflation (\lref[Section]{sec:set-cover}).

\ignore{The algorithm below follows the framework in \lref[Section]{sec:alg-overview}, which (though stated for uniform
inflation) easily extends to the case of non-uniform inflation. In particular the algorithm takes as input an \rcov
instance and bound $T$, and outputs first-stage solution $\fst\sse [m]$ and augmentation $\snd:\Omega\rightarrow
2^{[m]}$ such that:
\begin{enumerate}
 \item If $T\ge \Tstar$ then first-stage cost $b(\fst)\le \alpha_1\cdot b(\Phistar) + \alpha_2\cdot \Tstar$.
 %, where $\Phistar$ and $\Tstar$ (respectively)  denote the first-stage and second-stage cost of some optimal solution to the \rcov instance.
 \item For every $\omega\in \Omega$, the sets $\fst~\bigcup ~\snd(\omega)$ cover all elements in
 $\omega$, and  the cost $c\left(\snd(\omega)\right)\le \beta\cdot T$.
\end{enumerate}
}

\begin{algorithm}
  \caption{Algorithm for Cardinality Robust Set Cover with non-uniform inflation}
  \begin{algorithmic}[1]
    \STATE \textbf{input:} robust set-cover instance and bound $T$.
    \STATE \textbf{let} $\beta \gets 36\cdot \ln m$, and  $S\leftarrow \left\{v\in U \mid \mbox{ minimum {\bf $c$-cost} set
      covering $v$ has cost at least } \beta\cdot \frac{T}{k}\right\}$.

    \STATE \textbf{output} first stage solution \fst as the
    Greedy-Set-Cover($S$) under {\bf $b$-costs}.
   \STATE {\bf define} $\snd(\{i\})$ as the minimum {\bf $c$-cost} set covering $i$ if $i\in U\setminus S$, and $\emptyset$
   otherwise.
   %$\snd(\{i\})=\emptyset$ for $i\in S$.
     \STATE \textbf{output} second stage solution \snd where $\snd(\omega):=\bigcup_{i\in \omega}\snd(\{i\})$ for all
     $\omega\sse U$.
  \end{algorithmic}
\end{algorithm}

We will show that this algorithm is $\left( H_n,\, 12 H_n,\, 36 \ln m \right)$-discriminating. The following claim is
immediate.
\begin{claim}[Property A]\label{cl:gensc-2nd} For all $T\ge 0$ and $\omega\sse U$, the sets
$\fst\bigcup \snd(\omega)$ cover elements $\omega$; additionally if $|\omega|\le k$ then the cost $c(\snd(\omega))\le
\beta\, T$.
\end{claim}

\begin{theorem}[Property B']
  \label{th:main-gensc}
Assume $\beta \ge 36\cdot\ln m$. If $T\ge \Tstar$ then $b(\fst)\le H_n \cdot \left(b(\Phistar) + 12\cdot
\Tstar\right)$.
\end{theorem}
\begin{proof} We will show  that there is a {\em fractional} solution $\bar{x}$ for
covering $S$ with small {\bf $b$-cost}, at most $b(\Phistar) + 12\cdot \Tstar$,
  whence rounding this to an integer solution implies the theorem. For a
  contradiction, assume not: let every fractional set cover be
  expensive, and hence there must be a dual solution of large value.

Let $S'\sse S$ denote the elements that are {\em not}
  covered by the optimal first stage $\Phistar$, and let $\fprime \sse
  \F$ denote the sets that contain at least one element from $S'$. By
  the choice of $S$, all sets in $\fprime$ have {\bf $c$-cost} at least $\beta\cdot
  \frac{T}k\ge\beta\cdot \frac{\Tstar}k$. Define the ``coarse'' cost for
  a set $R \in \fprime$ to be $\newc_R = \ceil{\frac{c_R}{6\Tstar/k}}$.
  For each set $R\in\fprime$, since $c_R\ge \frac{\beta\Tstar}k \ge
  \frac{6\Tstar}k$, it follows that $\newc_R \cdot \frac{6\Tstar}k \in
  [c_R, 2\cdot c_R)$, and also that $\newc_R \ge \beta/6$.

Now consider the LP  for the set cover instance
  with elements $S'$ and sets $\fprime$ having the coarse costs $\newc$. Let
  $\{x_R\}_{R \in \fprime}$ be an optimal fractional solution; then \lref[Claim]{cl:main} applies directly to yield:
  \begin{equation}\label{eq:non-unif-sc}
  \sum_{R\in\fprime}
    \newc_R\cdot x_R \le 2\cdot k\end{equation}

Given the primal LP solution
  $\{x_R\}_{R \in \fprime}$ to cover elements in $S'$, define a fractional
  solution $z$ covering elements $S$ as follows: define $z_R = 1$ if
  $R\in\Phistar$, ${z}_R=x_R$ if $R\in \fprime$, and
  ${z}_R=0$ otherwise.  Since the solution $z$ contains $\Phistar$
  integrally, it covers elements $S \setminus S'$ (i.e. the portion of
  $S$ covered by $\Phistar$); since $z_R \geq x_R$ for all $R\in \fprime$, $z$
  fractionally covers $S'$.  Finally, the {\bf $b$-cost} of this solution is:
  $$b\cdot z = b(\Phistar)+ b\cdot x \le b(\Phistar)+ c\cdot x \le b(\Phistar)+ \frac{6\Tstar}k\cdot
  (\newc\cdot x)\le b(\Phistar) +12\cdot \Tstar,$$
where the second inequality uses $b\le c$, the next one is by definition of $\newc$ and the last inequality is
from~\eqref{eq:non-unif-sc}. Thus we have an LP solution
  of {\bf $b$-cost} $\Phistar + 12\Tstar$, and since the greedy algorithm is an
  $H_n$-approximation relative to the LP value, this
  completes the proof.
\end{proof}

Thus we obtain:
\begin{theorem}
  There is an $O(\log m+\log n)$-approximation for $k$-robust set cover with set-dependent inflations.
\end{theorem}

%---------------------------------------------------------------------

%------------------------------------------------------------------------------

\end{document}